\documentclass[pdftex,letterpaper]{vldb}

\newif\iffull
\fulltrue
\newif\ifnotfull
\iffull \notfullfalse \else \notfulltrue \fi

\usepackage[disable]{todonotes}

\usepackage[utf8]{inputenc}
\usepackage[T1]{fontenc}
\usepackage{comment}

\usepackage{tcolorbox}
\usepackage{balance}
\usepackage{empheq}
\usepackage{amssymb,amsmath,mathtools}
\usepackage{graphicx,tabularx}
\usepackage{subcaption}
\usepackage{paralist}
\usepackage{xspace}
\usepackage{color}
\usepackage{booktabs}
\usepackage{url}            
\usepackage{amsfonts}       
\usepackage{nicefrac}       
\usepackage{microtype}      
\usepackage[htt]{hyphenat}
\usepackage{bm}
\usepackage[ruled,vlined,linesnumbered]{algorithm2e}
\usepackage{verbatim}
\usepackage{braket}
\usepackage{paralist}
\usepackage{empheq}
\usepackage{makecell}
\usepackage{braket}
\usepackage{multirow}

\usepackage[nocompress]{cite}
\usepackage[pdfpagelabels,colorlinks,allcolors=blue]{hyperref}
\usepackage{tikz, pgfplots, tkz-euclide}
\usepgfplotslibrary{fillbetween}
\usetkzobj{all}
\usetikzlibrary{intersections}
\usetikzlibrary{calc}

\newcommand{\algname}{\texttt{GD}\xspace}

\let\oldnl\nl
\newcommand{\nonl}{\renewcommand{\nl}{\let\nl\oldnl}}

\newcommand{\gtodo}[1]{\todo[inline]{\small GY: #1}}

\newcommand{\dtodo}[1]{\todo[inline,backgroundcolor=green!10!white]{\small DA: #1}}
\newcommand{\hl}[1]{{#1}}

\newtheorem{theorem}{Theorem}[section]
\newtheorem{lemma}[theorem]{Lemma}
\newtheorem{proposition}[theorem]{Proposition}
\newtheorem{corollary}[theorem]{Corollary}

\newtheorem{definition}{Definition}[section]

\DeclareMathOperator*{\argmin}{\arg\!\min}

\renewcommand{\qed}{\nobreak \ifvmode \relax \else
	\ifdim\lastskip<1.5em \hskip-\lastskip
	\hskip1.5em plus0em minus0.5em \fi \nobreak
	\vrule height0.75em width0.5em depth0.25em\fi}

\newcommand{\eps}{\ensuremath{\varepsilon}}

\newcommand{\vx}{\mathbf x}
\newcommand{\vy}{\mathbf y}

\newcommand{\ve}{\mathbf e}
\newcommand{\vz}{\mathbf z}

\newcommand{\vlam}{\boldsymbol {\mathbf \lambda}}

\newcommand{\Oh}{{\ensuremath{\mathcal{O}}}}

\newcommand{\prob}[1]{\text{\textsc{#1}}\xspace}

\newcommand{\w}[1]{w^{(#1)}}
\newcommand{\wj}{\w j}
\newcommand{\wk}{\w k}
\newcommand{\h}[1]{h^{(#1)}}
\newcommand{\hj}{\h j}
\newcommand{\muj}{\mu^{(j)}}
\newcommand{\epsdef}{\eps}
\DeclareMathOperator*{\avg}{avg}

\newcommand{\eqbox}[1] {
\begin{tcolorbox}[
    standard jigsaw,
    opacityback=0,  
]
#1
\end{tcolorbox}
}

\title{Multi-Dimensional Balanced Graph Partitioning via Projected Gradient Descent\ifnotfull\titlenote{Full version: \url{http://arxiv.org/abs/1902.03522}}\fi}

\numberofauthors{3}
\author{
\alignauthor
  Dmitrii Avdiukhin\\
  \affaddr{Indiana University}\\
  \affaddr{Bloomington, IN}\\
  \email{davdyukh@iu.edu} \\
\alignauthor
  Sergey Pupyrev\\
  \affaddr{Facebook}\\
  \affaddr{Menlo Park, CA}\\
  \email{spupyrev@gmail.com}
\alignauthor
  Grigory Yaroslavtsev\\
  \affaddr{Indiana University}\\
  \affaddr{Bloomington, IN}\\
  \email{grigory@grigory.us}
}

\usepackage[switch]{lineno}

\begin{document}

\iffull
	\newcommand{\llabel}[1]{}
	\newcommand{\lref}[1]{}
\else
	\newcommand{\llabel}[1]{\hypertarget{llineno:#1}{\linelabel{#1}}}
	\newcommand{\lref}[1]{\hyperlink{llineno:#1}{\ref*{#1}}}
\fi
\newcommand{\rlabel}[2]{\llabel{#1_range_begin}#2\llabel{#1_range_end}\xspace}
\newcommand{\rref}[1]{\lref{#1_range_begin}-\lref{#1_range_end}\xspace}

\ifnotfull

\section*{Response to Reviewer Comments}

We would like to thank all reviewers for their insightful comments, which were very helpful for improving the presentation and expanding our experimental results and their discussion. We've fixed all typos and made extensive changes in this revision. Please, find our response below. 

\subsection*{Reviewer 1}
1. \textbf{Comment:} ``I do not find the definition of the multi-dimensional graph partitioning problem clear enough''.

\textbf{Response:} Thanks, we made extensive changes to the definition, giving several examples of different weight functions which can be used. Please, see lines~\rref{line:R1.1} 

\vspace{3pt}

2.  \textbf{Comment:} ``<...>there are several references to an extended version of the paper that is not included anywhere''; ``W1. The paper references an extended version several times. However, no such reference is included. This makes the paper not self-contained. Is this a simple mistake or is it done on purpose?''
    
\textbf{Response:} Thanks, this was a mistake. Link to the full version is now included into the submission: \url{http://arxiv.org/abs/1902.03522}.

\vspace{3pt}
    
3. \textbf{Comment:} ``W2. I find the paper not well-written and hard to follow. I think that the theory can be explained much better.''
    
\textbf{Response:} Thanks, the presentation was indeed a bit dense. We made extensive changes to improve presentation and specifically address your comment D1.
\vspace{3pt}
    
4. \textbf{Comment:} ``W2. <...>  The paper talks about multi-dimensional partitioning, however, it seems to focus exclusively on 2-d partitioning. Specifically, vertex+edge is the only 2-d partitioning considered. This relationship is actually not that easy to spot.''
    
\textbf{Response:} This is a great point. We have updated the paper to include experiments for $d = 3$ and $d = 4$. These experiments include balance on PageRank of vertices (as a proxy for importance) and on sums of degrees over neighbors. Unfortunately, due to space limitations these new experiments don't fit into the submission so we included them: (a) into the full version: \url{http://arxiv.org/abs/1902.03522}; (b) in a separate Section~\ref{sec:multd-experiments} in our response below.

\vspace{3pt}

5. \textbf{Comment:} ``D1. The transition from the MDBGP problem to the optimization formulation is not clear and not explained. There is no sum in the definition whereas a sum appears in the constraint of the optimization problem. The subsequent transition to Algorithm 1 which solves the optimization with gradient descent is also unclear.''

\textbf{Response:} We improved the presentation to clarify the relationship between MDBGP and its optimization formulation, see lines~\rref{line:R1.5}. To clarify the transition to Algorithm 1 we've prefaced the algorithm with an overview of each of its key steps in Section~\ref{sec:overview}.

\vspace{3pt}    

6. \textbf{Comment:} ``D2. I do not quite understand why three projections are considered in the paper. They do not seem to be compared against each other either theoretically nor experimentally.''

\textbf{Response:} Thanks, to clarify the relationship between different projection methods:
\begin{enumerate}
    \item We included a new Table~\ref{table:proj-properties} which compares theoretical properties of the three projection methods used. 
    \item We did, in fact, report an experimental comparison between the three methods in Figure~\ref{fig:proj} (and using different parameters for the exact projection). It might have been unclear, though, why we didn't plot Dykstra's projection performance. We've updated the caption to this figure to make it more clear why we didn't plot Dykstra's projection performance (same as exact so would be hard to see on the plot).
\end{enumerate}

\subsection*{Reviewer 2}

    1. \textbf{Comment:} ``W1. The authors omit a discussion and comparison against METIS's multi-constraint partitioning (which can at least be used for several billion-edge graphs).''

    \textbf{Response:} This is a great point, we've added an experimental comparison with METIS. See more in our response to D1 and D2.
    \vspace{3pt}
    
    2. \textbf{Comment:} ``W2. In my interpretation, the paper confirms that for Pregel-like systems, graph partitioning is not critical (based on improvements over the baseline Hash partitioning).''

    \textbf{Response:} Thanks, while the $10-30\%$ improvement that we show in the paper might seem small, we believe that specifically in critical large-scale computations which have to be performed frequently even minor improvements can lead to substantial reductions in costs and resource usage. See our response to D3 for more details.
    \vspace{3pt}
    
    3. \textbf{Comment:} ``W3. Several discussions in experiments are a bit shallow and need more extended discussions.''

    \textbf{Response:} Thank you, some experimental results indeed lacked good explanations. We have substantially revised all discussions of the experiments. We also addressed specifically your questions D5, D6 and D8.
    \vspace{3pt}
    
    4. \textbf{Comment:} ``D1: 1.2: There is at least 2 decades of work on multi-constraint graph partitioning. The one work that I'm aware of is by the METIS group (Multilevel Algorithms for Multi-Constraint Graph Partitioning by Karypis et al.) and I was able to find several others with a quick web search (e.g. Parallel static and dynamic multi-constraint graph partitioning, Parallel Multilevel Algorithms for Multi-constraint Graph Partitioning, PULP: Scalable Multi-Objective Multi-Constraint Partitioning for Small-World Networks). Despite citing Karypis paper, the authors' still say "To the best of our knowledge, MDBGP has not been studied before". Either I'm missing some subtle difference between the problems or this claim seems wrong.''
    
    \textbf{Response:} Thanks, this is a great point. We were indeed aware of the previous work by the METIS group and cited them in the introduction. However, due to some unfortunate editing error this didn't make it into our ``Previous work'' section. We've updated the previous work section with lines~\rref{line:R2.5-1} and~\rref{line:R2.5-2}. Our corrected claims are that:
    \begin{enumerate}
        \item Previous work on balanced graph partitioning (even one-dimensional) hasn't used gradient descent on a continuous relaxation.
        \item The literature on the multi-dimensional version is rather sparse (we've extended the list of references as per your suggestion) and the main tool available at the moment is METIS.
    \end{enumerate}
    \vspace{3pt}
    5. \label{comment:metis} \textbf{Comment:} ``D2: The METIS software can allow one to express with multiple constraints on graphs with about 10B edges on a single machine with a lot of RAM (I expect a TB to be enough). I think it should be possible to implement a distributed version of METIS to partition graphs with up to 100B edges too. So the authors should compare their algorithm against multi-constraint METIS. Both on the 2 constraint case as well as the >2-constraint experiments (for which the current paper currently has no baselines).''
    
    \textbf{Response:} 
    Great point. Unfortunately, we couldn't make any of the publicly available parallel/distributed version of METIS (ParMETIS and MT-METIS, most recently released in 2013) run on our largest graphs due to its high memory utilization. In the paper the largest graph we report is FB-800B, which has 800B edges. We don't believe that any version of METIS can handle this in a reasonable time. It has also been observed previously in the literature that METIS doesn't scale well to large graphs even for $d = 1$ and hence every large company has developed its own set of large-scale graph partitioning tools (see e.g. the paper from Google~\cite{ABM16} and discussion in an earlier paper from Facebook on Spinner~\cite{MLLS17}). Note that the advantage we are getting in speed over tools like METIS is not only due to our own algorithm's efficient implementation but also due to the highly optimized Giraph platform (publicly available) that it is implemented on. Since METIS isn't implemented in Giraph it can't take advantage of it.
    
    
    We included experiments comparing our algorithm with multi-constraint METIS on graphs on which we could make it run in reasonable time. The experiments show that METIS achieves much worse balance than GD for partitioning under multiple constraints. Furthermore, in almost all cases cases our algorithm outperforms METIS with respect to other important performance metrics, such as cut size, memory usage and time. Most importantly, this effect becomes most pronounced as the size of data and the number of dimensions grows (e.g. our performance is best for all parameters, on the largest dataset orkut for $d = 3$ and $d = 4$), see Table~\ref{table:gd_vs_metis}.
    Unfortunately, due to space limitations these new experiments don't fit into the submission so we included them: (a) into the full version: \url{http://arxiv.org/abs/1902.03522}; (b) in a separate Section~\ref{sec:multd-experiments} in our response below.
     \vspace{3pt}
  
    6. \textbf{Comment:} ``D3: Importance of partitioning: I am under the impression and the modest runtime improvement results in this paper corroborate this impression that Hash is, despite being the simplest partitioning algorithms, quite efficient. That is because it at least gets close to perfect and vertex and edge imbalance despite not localizing the edges. A discussion on this could be good (and perhaps add this as a limitation of the current work (and in general any work on partitioning)). ''
    
    \textbf{Response:} This is a great question. We have two figures which illustrate the improvements:
    \begin{enumerate}
    \item Inf Figure~\ref{fig:hist} we show performance improvements for PageRank ($\approx 25\%$ compared to Hash).
    \item We also show performance improvements for other typical Giraph jobs (Mutual Friends, Connected Components, Hypergraph Clustering) in Figure~\ref{fig:digraph}, where vertex+edge balancing leads to a $10\%-30\%$ speedup.
    \end{enumerate}
    These numbers may seem relatively small, but they mean that we are able to handle the same workload with approximately $10\%-30\%$ less machines and electricity.
    For critical workloads in a large software company this can lead to very substantial savings.
\vspace{3pt}    
    
   7. \textbf{Comment:} ``D5: Any insights into why GD outperforms other techniques more on the FB graph and not on the other graphs. ''
    
    \textbf{Response:} This is a really great question. We assume that you are referring to our observation in lines~\rref{line:R2.7-reviewer} that GD achieves a bigger advantage over BLP on the FB graphs than on other graphs. The main difference between FB graphs and publicly available graphs is size. So the main reason for this advantage is poor performance of exsting local-search based methods (like BLP) on large graphs in the multi-dimensional case. This is most obvious in Figure~\ref{fig:quality_large} for $k = 128$ as one can see that GD is gaining a larger advantage over BLP as the size of the graph grows (3B $\rightarrow$ 80B $\rightarrow$ 400B).
    We also added this discussion in the revision in lines~\rref{line:R2.7}. 
        \vspace{3pt}

    8. \textbf{Comment:} ``D6: Section 4.1: Why do the authors report numbers on Orkut only for the d > 2 experiments and GD and not the others? Because of this limitation it raises the question: was GD not working well for d > 2 on datasets other than Orkut?''
    
    \textbf{Response:} Thanks, this is a great point. We addressed this together with D2 by including a new set of experiments which compares GD against METIS on three datasets from SNAP (LiveJournal, orkut, sx-stackoverflow). These experiments are shown in Table~\ref{table:gd_vs_metis}. Please, see our response to D2 for a discussion of these new results. 
        \vspace{3pt}

    9. \textbf{Comment:} ``D7: Overview of Dykstra's method is not given. To make the paper self-contained, I think the authors should give the necessary background on this method.''
    
    \textbf{Response:} Thanks, Dykstra's projection algorithm is rather standard. Given the space constraints in the revision we included a link to the corresponding Wikipedia article: \url{https://en.wikipedia.org/wiki/Dykstra%27s_projection_algorithm}.
        \vspace{3pt}

    10. \textbf{Comment:} ``D8: Interpretation of the results: I was not satisfied with the experiment on the different projection methods: The authors conclude with this: "The results show that the exact projection algorithm performs the best if sufficiently large imbalance is allowed, but the alternating projections algorithm can be often used to achieve similar performance." But this is not interesting and does not seem to teach the community much. The interesting question is why? Without any insights into the behavior of these algorithms, I find the experiments unsatisfying.''
    
    \textbf{Response:}  This is a good question. We've clarified the claim slightly: "The results show that the exact projection algorithm with sufficiently large allowed imbalance leads to the best performance. However, the alternating projections algorithm can be often used to achieve similar performance."
    
    Larger allowed imbalance leads to more flexibility in the search space of solutions, which allows our algorithm to better optimize edge locality. Despite not computing the exact projection alternating projection often computes a point which is close enough to it.
    For updated discussion in the paper please see lines~\rref{line:R2.10}.
        \vspace{3pt}

    11. \textbf{Comment:} ``However, edge-based graph partitioning can still result in performance regressions [2,37]. => Can you elaborate on why this happens, that is the shortcomings of edge partitioning?''
    
    \textbf{Response:} This is a good question. We had some discussion of why edge-based partitioning alone is not enough in the next paragraph (lines~\rref{line:R2.11} of the submission). The intuition is that some part of the cost of running a job on each worker depends on the number of its local vertices rather than edges. For example, for PageRank the vertex-based cost is associated with serialization. Please, let us know what other explanation might help.

\subsection*{Reviewer 3}

    1. \textbf{Comment:} ``W1. The motivation of the research problem can be well discussed in the introduction.''
    
    \textbf{Response:} Thanks, we've made multiple edits to address this. Please, see more detailed response to your comment D1. 
    
    \vspace{3pt}
    
    2. \textbf{Comment:} ``W2. Some theoretical analysis about the algorithm accuracy can be included in the paper.'' 
    
    \textbf{Response:} Good point. Unfortunately, this problem is NP-hard, and is hard to approximate, please see our response to D2 for more details. 
    
    \vspace{3pt}
    
    3. \textbf{Comment:} ``W3. Various type of datasets should be reported in the experiment.''
    
    \textbf{Response:} Thank you for the great suggestion. We added experiments on the additional dataset, please see our response to D3.
    
    \vspace{3pt}

    4. \textbf{Comment:} ``D1. The paper focuses on the multi-dimensional balanced graph partitioning. More discussion should be given in the introduction about the motivation of the term ``multi-dimensional''. Some examples or applications can be given here.''
    
    \textbf{Response:} Thanks, we have expanded to include more examples and applications of multi-dimensional partitioning in lines~\rref{line:R3.4} 
    
    \vspace{3pt}
    
    5. \textbf{Comment:} ``D2. Can the algorithm return the exact result? Some theoretical analysis about the accuracy of the algorithm can be given.''
    
    \textbf{Response:} Good question, but unfortunately the algorithm can't return an exact or even approximate result even in the 1-dimensional case due to NP-hardness of both exact and approximate balanced graph partitioning. This is why accuracy of all scalable practical algorithms for this problem so far have only been analyzed empirically and all such algorithms were based on various flavors of combinatorial local search. 
    
    \vspace{3pt}
    
    6. \label{comment:dataset} \textbf{Comment:} ``D3. All selected datasets are social networks in the experiment. Some other type of datasets should be included.''
    
    \textbf{Response:} This a great point. 
    There are two main reasons for testing on social networks: 1) they are some of the largest graphs available and used in applications, 2) multidimensional partitioning helps achieve substantial speedups when running typical social network processing jobs, as we illustrate in the paper.
	We also revised the paper to include results on \texttt{sx-stackoverflow}~-- largest SNAP graph which is not a social network.
    Unfortunately, due to space limitations these new experiments don't fit into the submission so we included them: (a) into the full version: \url{http://arxiv.org/abs/1902.03522}; (b) in a separate Section~\ref{sec:extra-data-experiments} in our response below.
    
    \vspace{3pt}
        
   7.  \textbf{Comment:} ``D4. It should be better to report the running time of the proposed algorithm on the real-world datasets compared with the existing methods.''
    
    \textbf{Response:} Thanks, this is a great suggestion. In our experiments all balanced graph partitioning methods which scale to graphs with 800B edges (GD, SHP, Spinner, BLP) had comparable running times (within the same order of magnitude). Given that, we are not sure how meaningful the exact running times are as each algorithm, including GD, allows some room for non-asymptotic optimization. This is why we only report the running time of our algorithm GD in Figure~\ref{fig:proj}.
    
    \vspace{3pt}
    
    8. \textbf{Comment:} ``D5. It should be better to give a running example for the algorithm for the clearance of the paper.''
    
    \textbf{Response:} Thank you, this is a good suggestion. We would have really liked to implement it ourselves, but, unfortunately, due to continuous nature of he algorithm it's hard to give simple examples which illustrate the performance of the algorithm well. Our algorithm performs non-convex optimization in $n$-dimensional space where $n$ is the number of vertices in the graph, which makes it difficult to illustrate.  This is unlike existing combinatorial methods, which can be illustrated easily using small examples. Since our algorithm is based on gradient descent, we believe that convergence behavior, shown in our experiments, is a good compromise.
\dtodo{Show a few iterations on a path graph.}

\ifnotfull
\setcounter{section}{-1}
\fi

\section{Additional experiments}\label{app:more_experiments}

In this section we show experiments for $d > 2$ and compare performance of \algname with METIS.
We also show experiments on dataset \texttt{sx-stackoverflow}~-- the largest SNAP graph which is not a social network.

\subsection{Multi-dimensional experiments}\label{sec:multd-experiments}

We performed experiments for $d = 3$ and $d = 4$ to illustrate the performance of our algorithms in the multi-dimensional case. We remark that our algorithm can handle higher dimensions as well, but public weight data for large enough graphs is hard to find. For these multidimensional experiments in addition to balancing on the number of vertices and edges we also balance based on the following additional vertex weights:
\begin{compactitem}
	\item \emph{Pagerank}. We use Pagerank to model activity level of a node. High Pagerank likely means that the vertex is accessed often, and therefore balancing on Pagerank can be beneficial for load balancing purposes.
	\item \emph{Sum of neighbor degrees}. We also use the sum of degrees over neighbours of a vertex as a weight function. We choose the sum of neighbor degrees as a proxy for the size of the 2-hop neighborhood of a vertex, which is computationally expensive to compute for very large graphs.
\end{compactitem}	
	
\dtodo{We should say that we compare with METIS}
The results are presented in Table~\ref{table:gd_vs_metis}.
They indicate that METIS achieves poor balance for multiple constraints and that \algname outperforms METIS by almost all parameters in most cases (better results shown in bold).
METIS was given allowed imbalance of $0.5\%$.

\begin{table*}[!htb]
	\centering
\begin{tabular}{|l|c|cc|cc|cc|}
	\cline{3-8}
	\multicolumn{1}{c}{} & \multicolumn{1}{c|}{} & \multicolumn{2}{c|}{LiveJournal} & \multicolumn{2}{c|}{orkut} & \multicolumn{2}{c|}{sx-stackoverflow}	\\ \cline{3-8}
	\multicolumn{1}{c}{} & \multicolumn{1}{c|}{} & \algname & METIS & \algname & METIS & \algname & METIS \\ \hline
	\multirow{4}{*}{\shortstack[l]{$d = 2$: balance on vertices \\ and degrees}}
	& Locality, $\%$ & $\mathbf{91.71}$ & $93.74$ & $\mathbf{88.36}$ & $86.52$ & $75.82$ & $\mathbf{80.41}$ \\ \cline{2-8}
	& $\max$ imbalance, $\%$ & $\mathbf{0.04}$ & $0.5$ & $\mathbf{0.02}$ & $0.7$	& $\mathbf{0.04}$ & $0.6$ \\ \cline{2-8}
	& Memory, MB & $\mathbf{2635}$ & $4085$ & $\mathbf{4673}$ & $10259$ & $\mathbf{1587}$ & $4113$ \\ \cline{2-8}
	& Time, s & $117$ & $\mathbf{44}$ & $203$ & $\mathbf{92}$ & $68$ & $\mathbf{55}$ \\ \hline
	\multirow{4}{*}{\shortstack[l]{$d = 3$: balance on vertices, \\ degrees and \\ sum of neighbor degrees}}
	& Locality, $\%$ & $\mathbf{88.74}$ & $73.36$ & $\mathbf{89.55}$ & $62.1$ & $\mathbf{76.8}$ & $60.09$ \\ \cline{2-8}
	& $\max$ imbalance, $\%$ & $\mathbf{0.05}$ & $30$ & $\mathbf{0.02}$ & $1.6$ & $\mathbf{0.1}$ & $6.5$ \\ \cline{2-8}
	& Memory, MB & $\mathbf{2711}$ & $4802$ & $\mathbf{4697}$ & $12271$ & $\mathbf{1627}$ & $4985$ \\ \cline{2-8}
	& Time, s & $140$ & $\mathbf{66}$ & $\mathbf{196}$ & $303$ & $\mathbf{76}$ & $131$ \\ \hline
	\multirow{4}{*}{\shortstack[l]{$d = 4$: balance on vertices, \\ degrees, \\ sum of neighbor degrees \\ and pagerank}}
	& Locality & $\mathbf{87.93}$ & $74.36$ & $\mathbf{75.58}$ & $65.08$ & $77.04$ & $\mathbf{78.54}$ \\ \cline{2-8}
	& $\max$ imbalance, $\%$ & $\mathbf{0.5}$ & $38$ & $\mathbf{2.7}$ & $20$ & $\mathbf{0.4}$ & $3.8$ \\ \cline{2-8}
	& Memory, MB & $\mathbf{2939}$ & $4839$ & $\mathbf{4896}$ & $12294$ & $\mathbf{1754}$ & $5013$ \\ \cline{2-8}
	& Time, s & $227$ & $\mathbf{66}$ & $\mathbf{240}$ & $297$ & $\mathbf{88}$ & $142$ \\ \hline
\end{tabular}
\caption{Comparison of \algname with METIS for multidimensional experiments. The results show that for high-dimensional balanced partitioning METIS can't guarantee balance.
Better results shown in bold. In most cases \algname outperforms METIS by almost in edge locality, imbalance, memory usage and/or time.}
\label{table:gd_vs_metis}
\end{table*}

\subsection{Experiments on Q\&A data}\label{sec:extra-data-experiments}

In this section we present experimental results on SNAP graph \texttt{sx-stackoverflow}, containing $2\,601\,977$ vertices and $28\,183\,518$ edges after removing duplicate edges.
Unlike other graphs presented in this paper, this one is not a social network.
The experiments show that performance of \algname on this graph is similar to other social network graphs included in the paper.

\begin{figure}[!htb]
	\centering
	\begin{subfigure}[b]{0.23\textwidth}
		\centering
		\includegraphics[width=\textwidth]{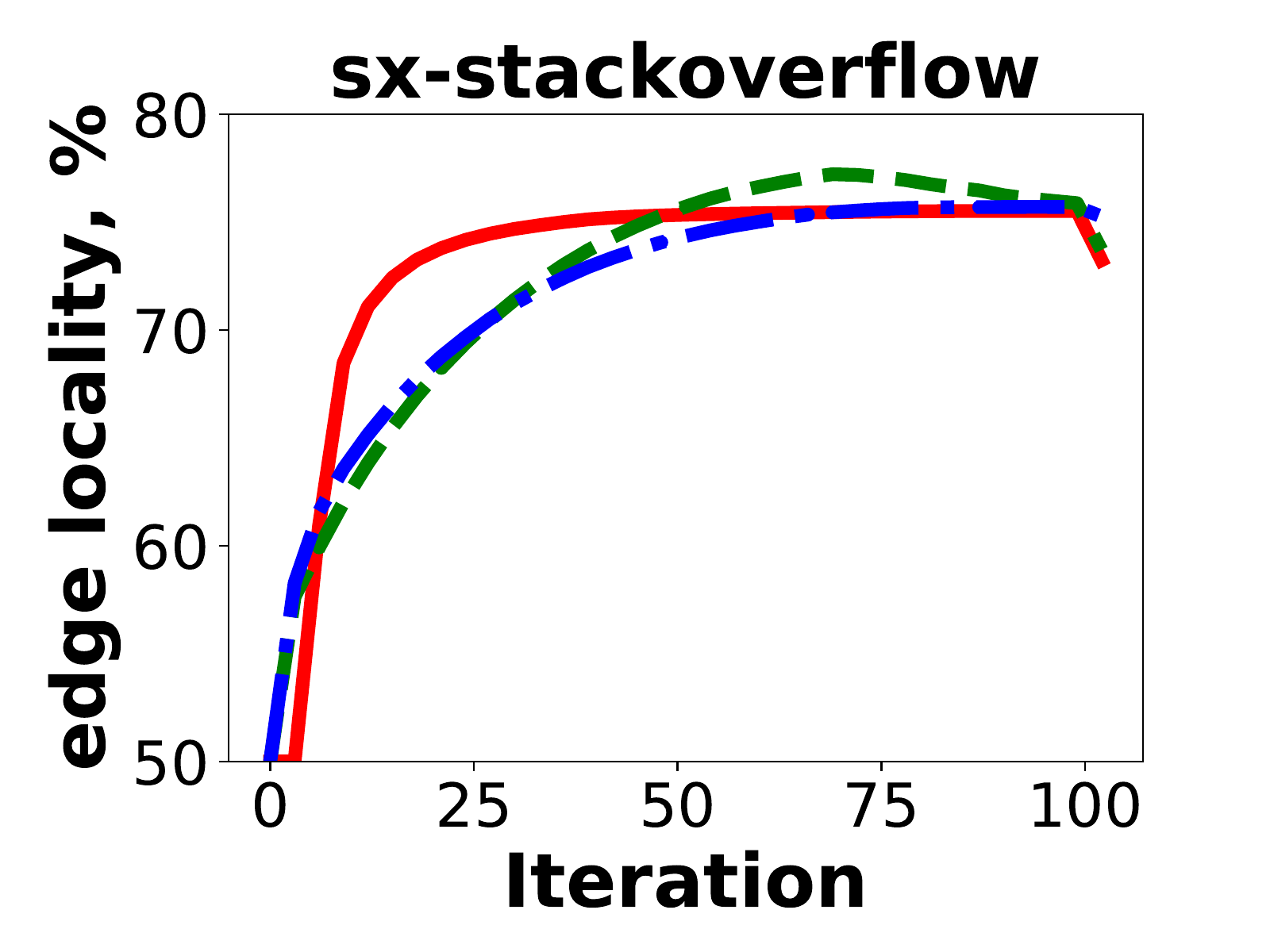}
	\end{subfigure}
	\begin{subfigure}[b]{0.23\textwidth}
		\centering
		\includegraphics[width=\textwidth]{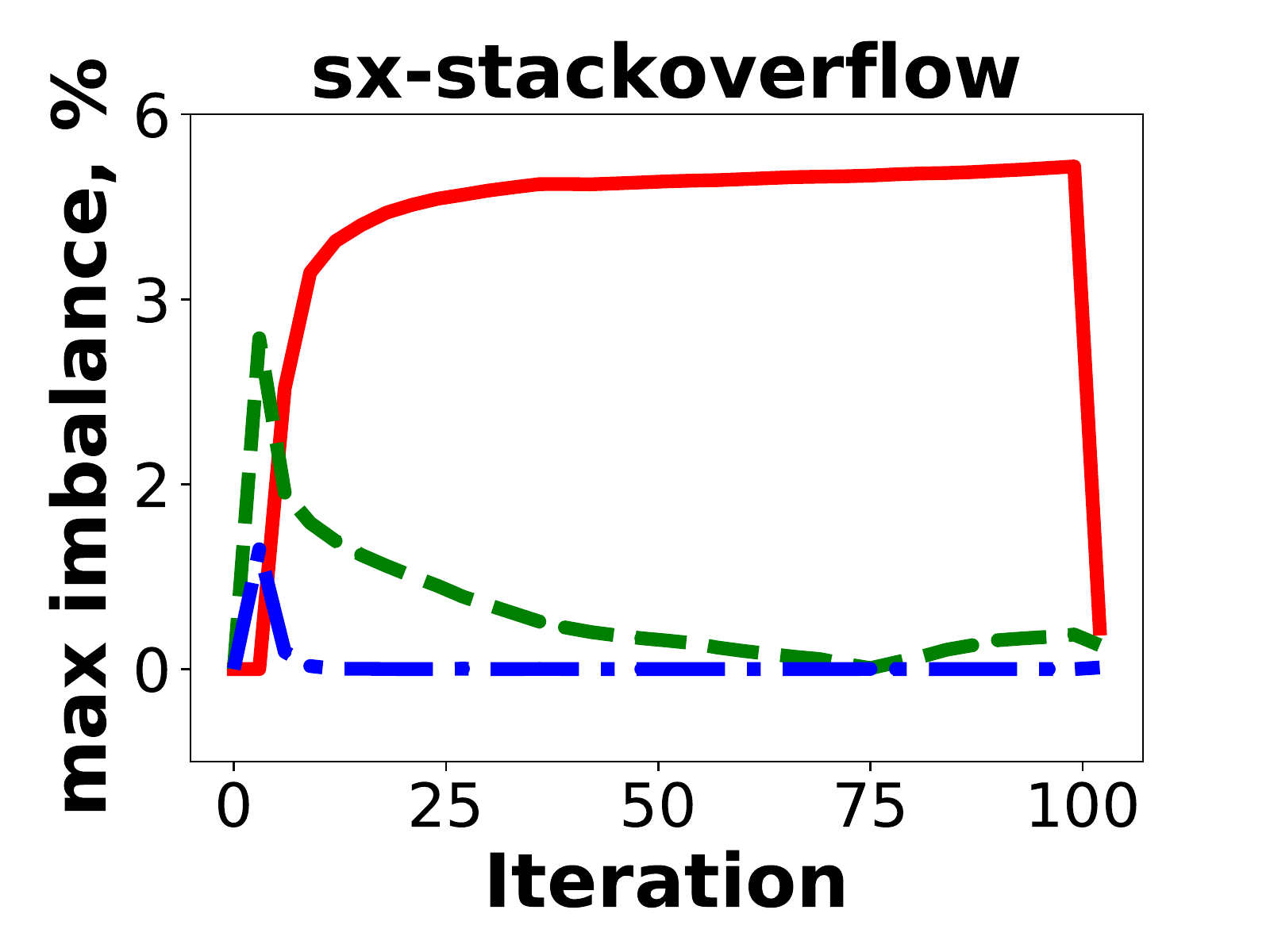}
	\end{subfigure}
	\begin{subfigure}[b]{0.23\textwidth}
		\centering
		\includegraphics[width=\textwidth]{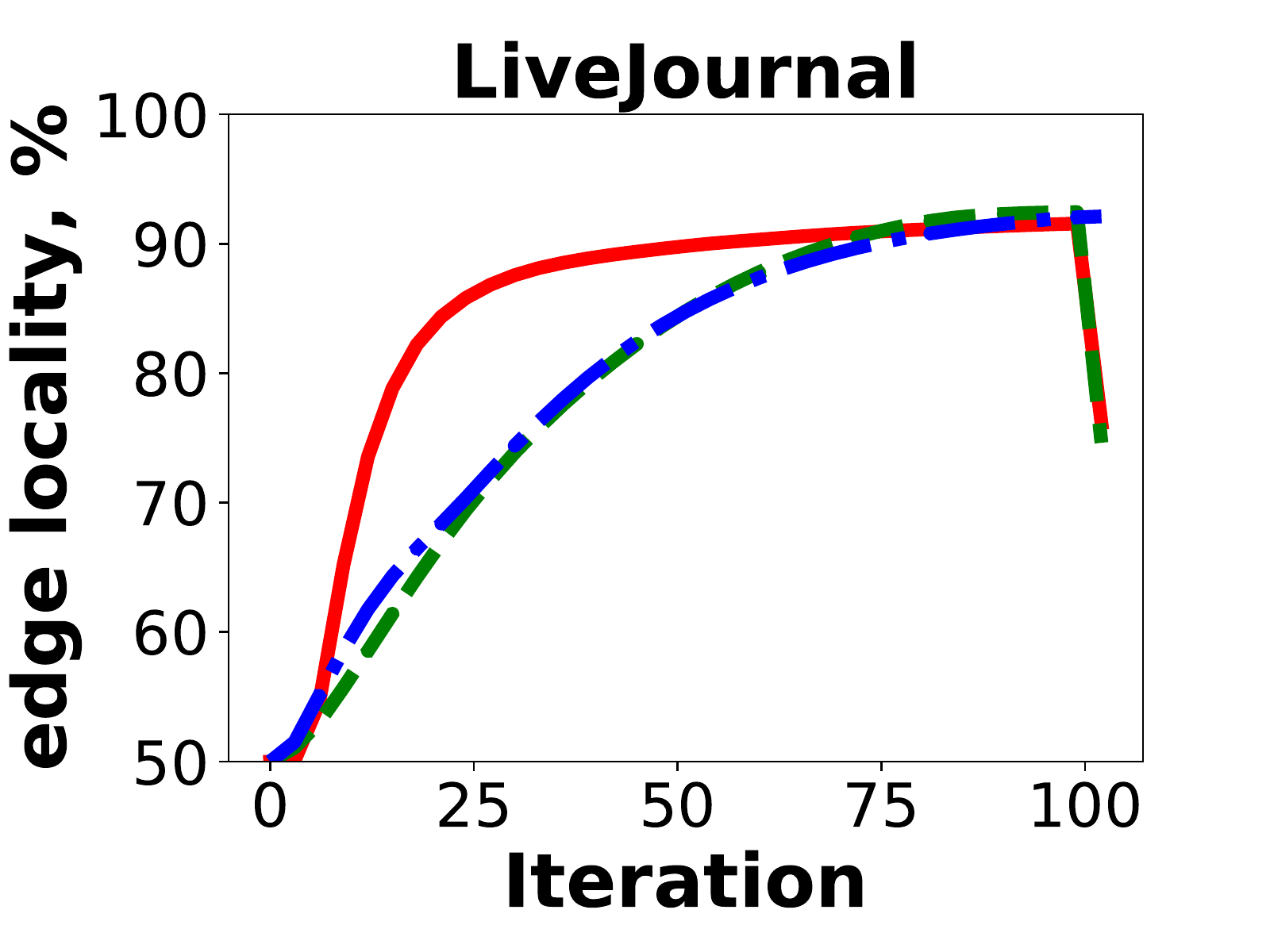}
	\end{subfigure}
	\begin{subfigure}[b]{0.23\textwidth}
		\centering
		\includegraphics[width=\textwidth]{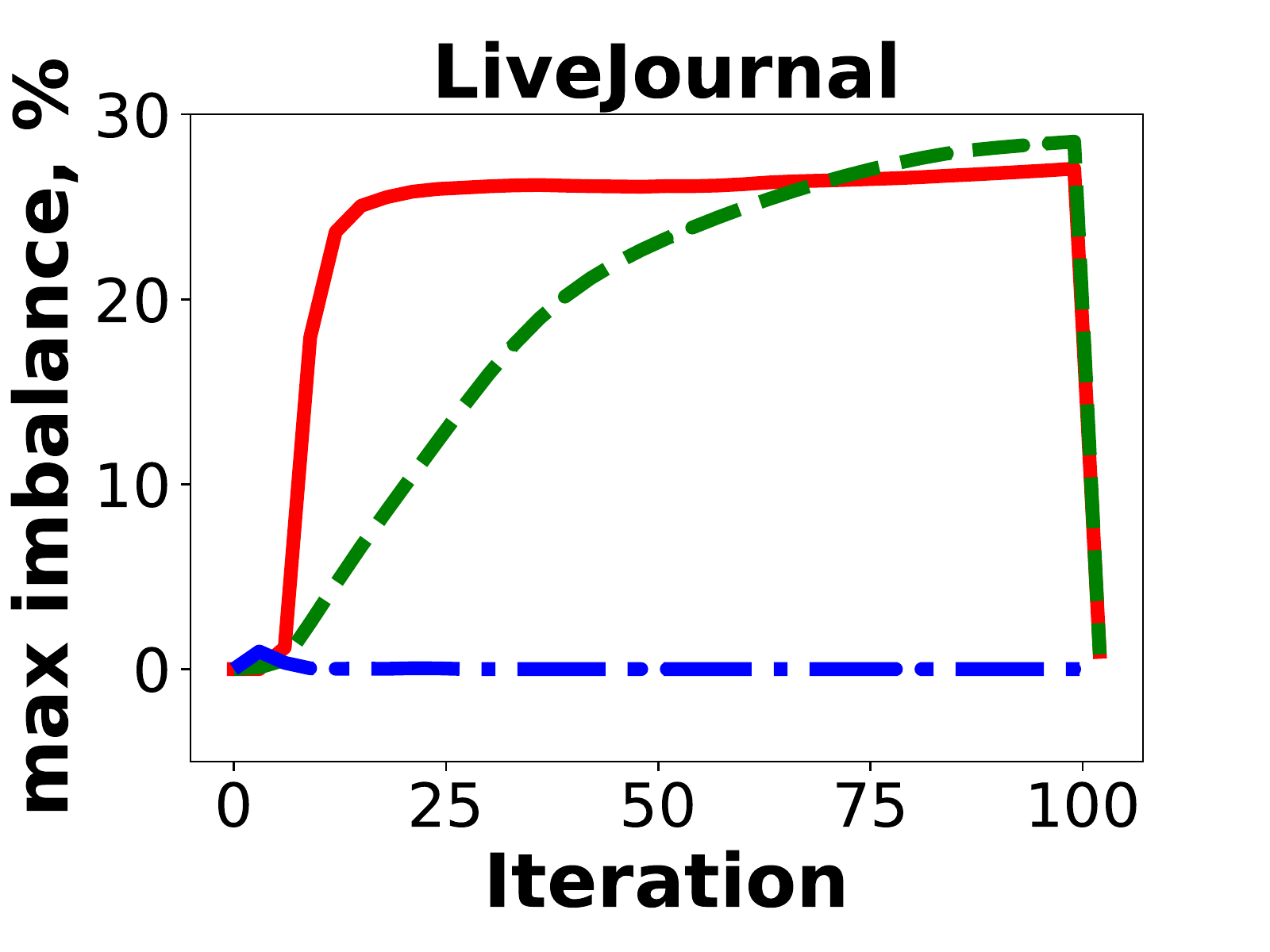}
	\end{subfigure}
	\begin{subfigure}[b]{0.46\textwidth}
	\centering
	\includegraphics[width=\textwidth]{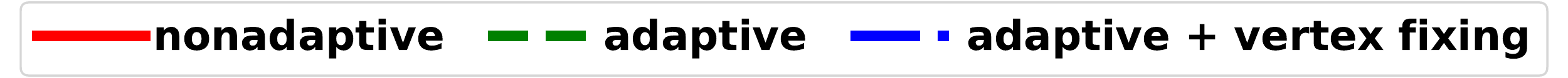}
\end{subfigure}
	\caption{Quality and imbalance comparison of \algname 1) without adaptive step size, 2) with adaptive step size and 3) with adaptive step size with vertex fixing as in Figure~\ref{fig:adapt}. The left side shows edge locality and the right side~-- maximum imbalance over all dimensions.
	While \algname with adaptive step size behaves better on \texttt{sx-stackoverflow} than on \texttt{LiveJournal}, \algname with adaptive step size and vertex fixing still results in slightly better locality and better balance during execution.}
	\label{fig:adapt_so}
\end{figure}
\dtodo{Fix plot captions}

\begin{figure}[!htb]
	\centering
	\begin{subfigure}[b]{0.23\textwidth}
		\centering
		\includegraphics[width=\textwidth]{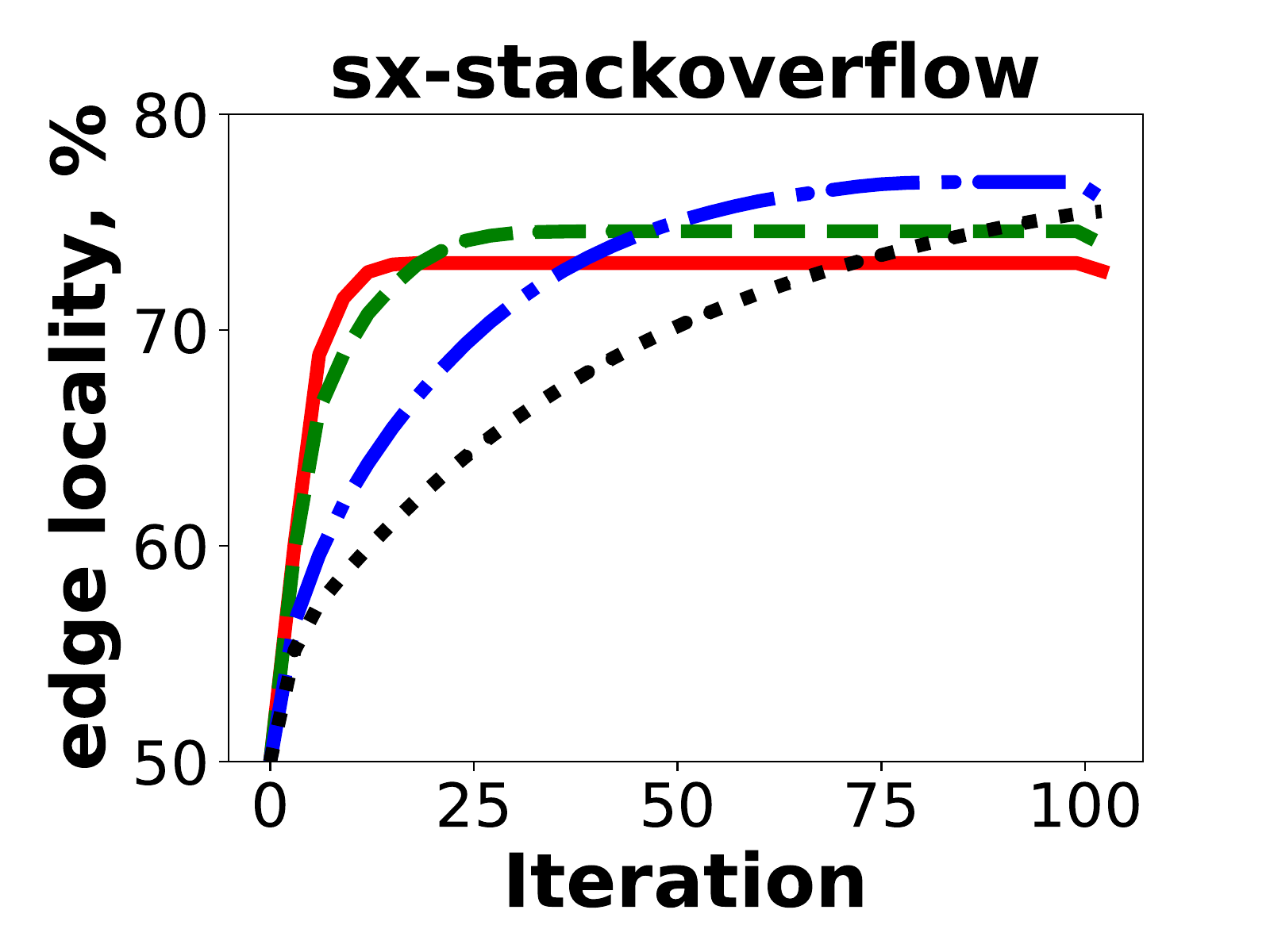}
	\end{subfigure}
	\begin{subfigure}[b]{0.23\textwidth}
		\centering
		\includegraphics[width=\textwidth]{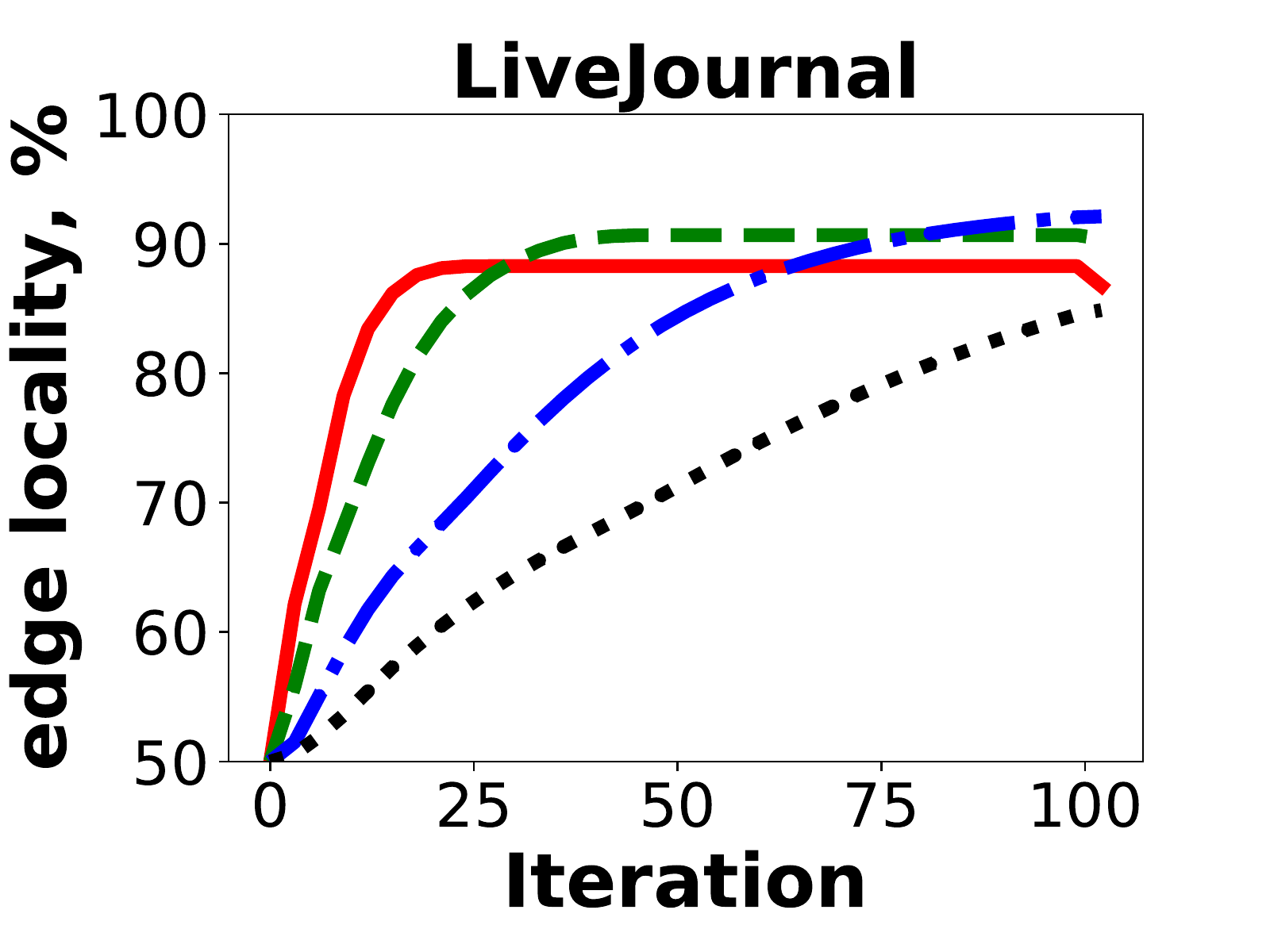}
	\end{subfigure}
	\begin{subfigure}[b]{0.46\textwidth}
	\centering
    	\includegraphics[width=\textwidth]{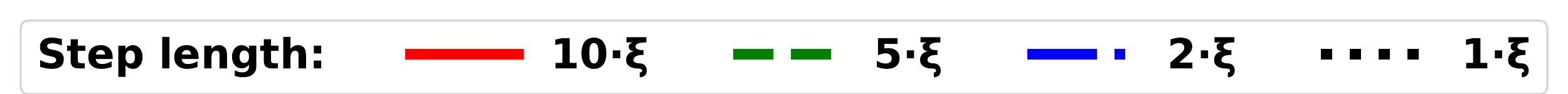}
\end{subfigure}
	\caption{Comparison of step choices for \algname with fixed step length, as in Figure~\ref{fig:step}. \hl{\algname runs for $100$ iterations and $\xi = \sqrt{n}/100$.} For \texttt{sx-stackoverflow} the rate of convergence is faster than for \texttt{LiveJournal}, and the advantage of step length choice of $2 \cdot \xi$ is more prominent}
	\label{fig:step_so}
\end{figure}
\dtodo{Say what $\xi$ is}
\begin{figure}[!htb]
	\centering
	\begin{subfigure}[b]{0.23\textwidth}
		\centering
		\includegraphics[width=\textwidth]{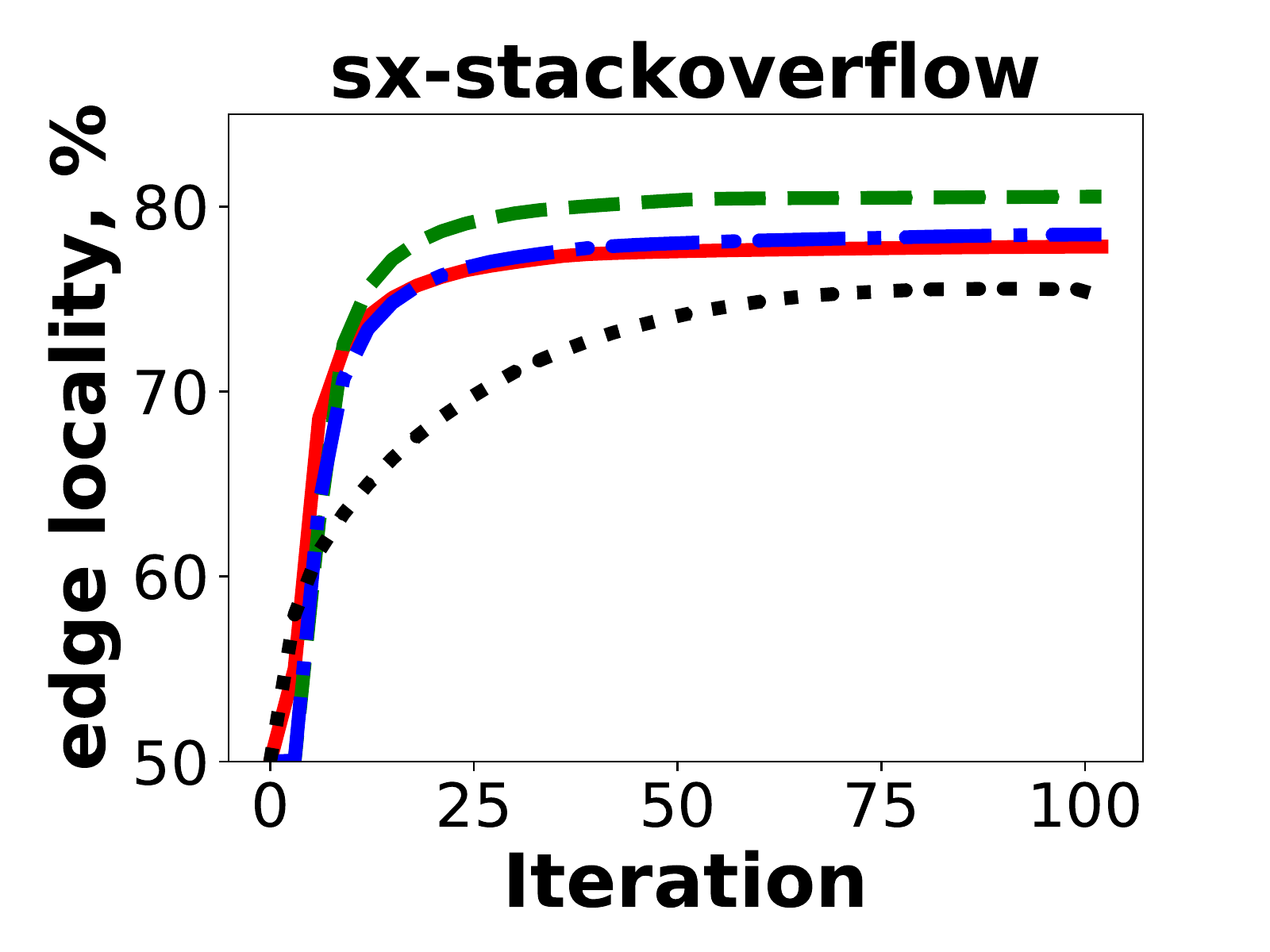}
	\end{subfigure}
    \begin{subfigure}[b]{0.23\textwidth}
      \centering
      \includegraphics[width=\textwidth]{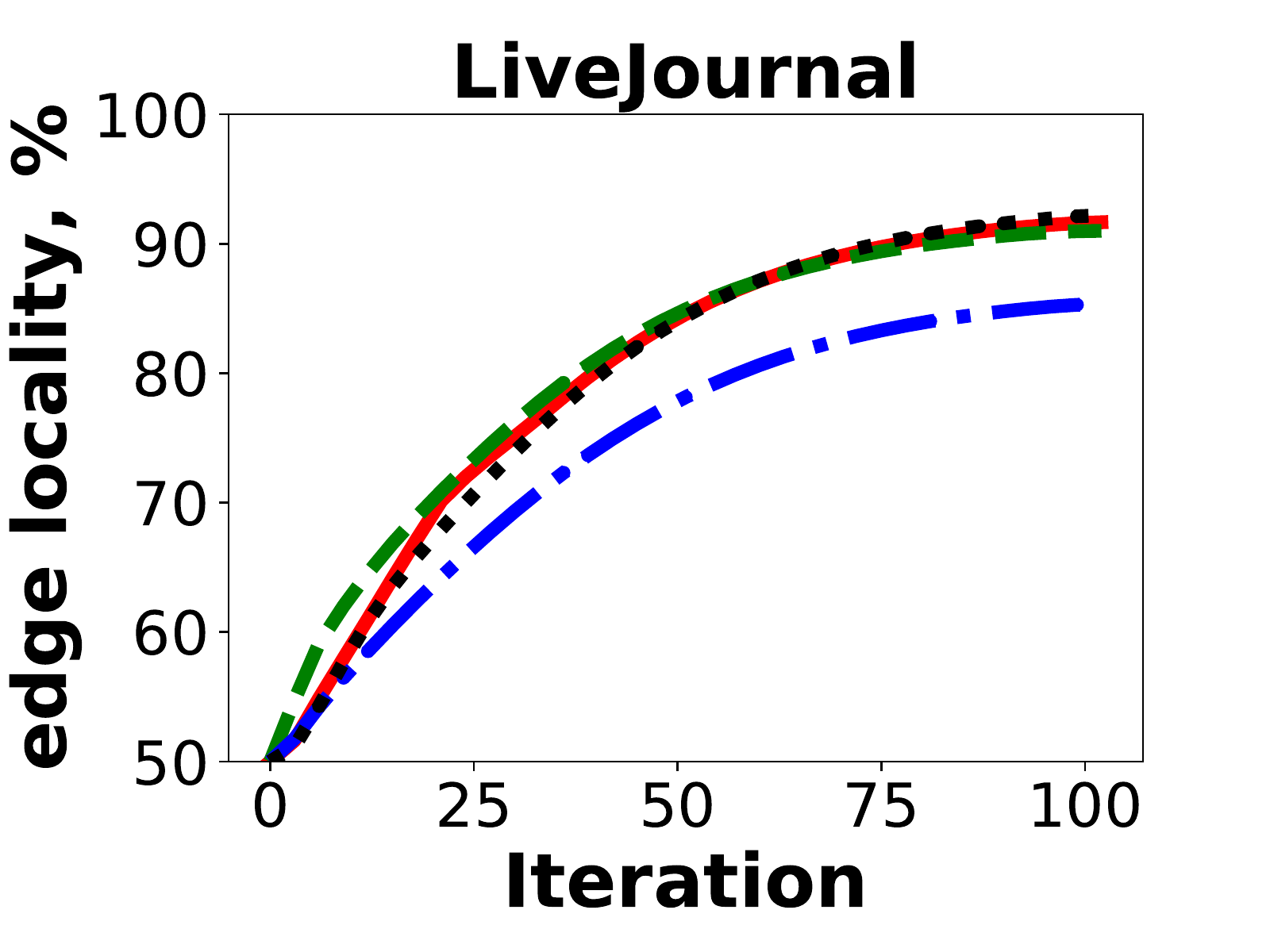}
    \end{subfigure}
	\begin{subfigure}[b]{0.46\textwidth}
		\centering
		\includegraphics[width=\textwidth]{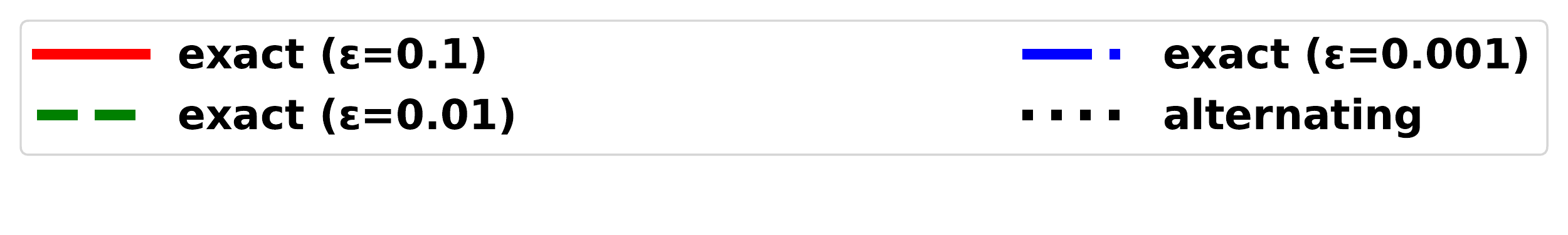}
	\end{subfigure}
    \caption{Quality comparison of \algname with various projection methods, as in Figure~\ref{fig:proj}. \hl{For \texttt{sx-stackoverflow} the rate of convergence of exact projection methods is faster, and the resulting locality is better.}}
	\label{fig:proj_so}
\end{figure}

	\newpage
    {\color{white} x}
	\newpage
    {\color{white} x}
	\newpage
\fi

\sloppy
\date{}

\ifnotfull
\linenumbers
\fi

\setcounter{page}{1}

\maketitle 

\begin{abstract}
Motivated by performance optimization of large-scale graph processing systems that distribute the graph across multiple machines, we consider the balanced graph partitioning problem. 
Compared to most of the previous work, we study the \emph{multi-dimensional} variant when balance according to multiple weight functions is required. As we demonstrate by experimental evaluation, such multi-dimensional balance is essential for achieving performance 
improvements for typical distributed graph processing workloads.

We propose a new scalable technique for the multi-dimensional balanced graph partitioning problem. The method is based on applying randomized 
projected gradient descent to a non-convex continuous relaxation of the objective. We show how to implement the new algorithm efficiently in both 
theory and practice utilizing various approaches for the projection step.
Experiments with large-scale graphs with up to 800B edges indicate that 
our algorithm has superior performance compared with the state-of-the-art approaches.
\end{abstract}

\section{Introduction}
\label{sec:intro}
Distributed graph processing systems have been widely adopted in recent years to enable analysis and knowledge extraction from large-scale graphs.
Systems such as Giraph~\cite{Chi11}, GraphX~\cite{GXDCFS14}, GraphLab~\cite{LBGGKH12}, and PowerGraph~\cite{GLGBG12}
  allow users to use a vertex-centric model for applications which can be executed on a cluster of worker nodes.
In this setting, each worker node operates on a subset of the input graph and communicates with other workers by sending messages.
The process of splitting the input graph into these subsets, also known as graph partitioning,
  is essential for optimizing performance of such systems~\cite{GLGBG12,VLSG17,GHCIE17,AKCV18}.

Created partitions have a significant impact on the communication between different workers and the resource usage of individual workers.
In order to maximize the processing speed, the partitions should largely be independent to minimize communication.
At the same time, computation executed on each partition should take approximately the same amount of processing time,
  as the overall performance depends on the slowest worker.
These constraints give rise to the \textsc{Balanced Graph Partitioning} problem
  whose goal is to divide the vertices of a graph into a given number of (approximately) equal size components while minimizing the resulting edge cut.
\textsc{Balanced Graph Partitioning}  is a classic and thoroughly studied problem from both theoretical and practical points of view~\cite{BS13,BMSSS16}.
In the context of distributed graph processing, the problem is typically studied in two variants.

In the \emph{vertex partitioning} model, each worker machine is assigned an equal number of vertices with the goal of minimizing
the number of cross-machine edges. Since messages are usually sent between adjacent vertices, storing tightly connected subgraphs on the 
same worker can reduce communication and hence running times of jobs. 
It has however been observed that this strategy does not lead to equally loaded partitions for real-world graphs with power law
degree distribution~\cite{GLGBG12}. Graph partitioning algorithms tend to colocate high-degree vertices and corresponding partitions take much longer to process, resulting in longer execution time overall.

\begin{figure*}[t]
	\centering
	\includegraphics[width=0.8\textwidth]{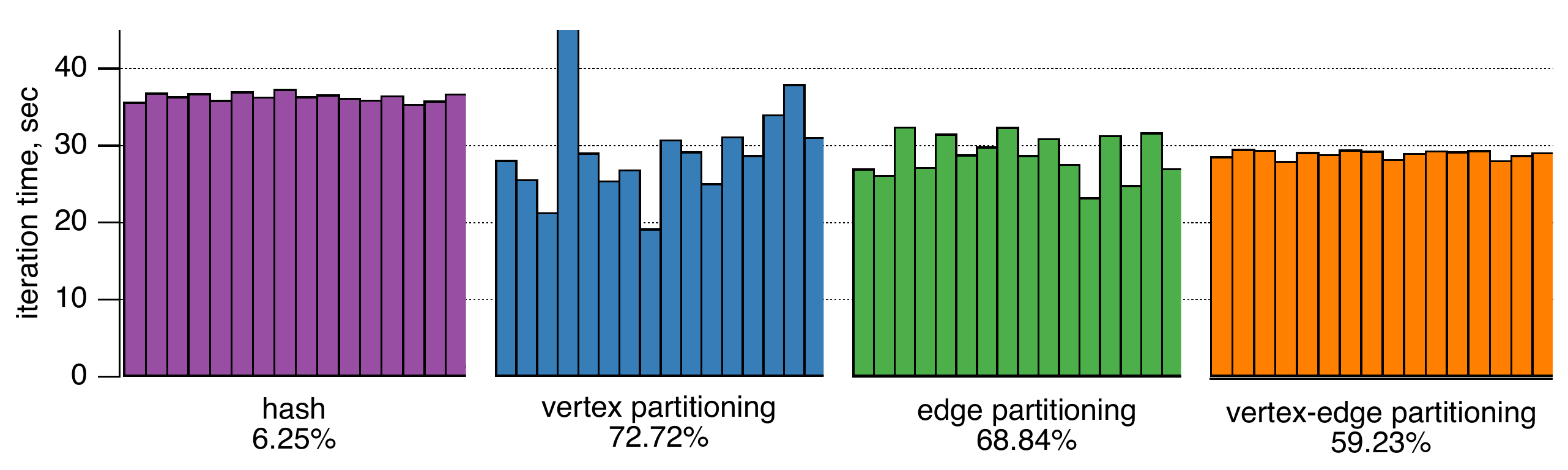}
	\caption{The running time of an iteration of \texttt{Page Rank} on a Giraph cluster of $16$ worker machines using
		various graph partitioning strategies. The numbers indicate the average percentage of local (uncut) edges per worker, which is 
		proportional to the number of local messages for the distributed graph processing workload. Vertex-edge partitioning achieves approximately $25\%$ iteration time improvement compared to hash.}
	\label{fig:hist}
\end{figure*}

The \emph{edge partitioning} model has been suggested to alleviate the above imbalance problem~\cite{GLGBG12,LGHC17}.
In this model the goal is to partition the graph so that the number of edges in every component is the same, while the number of incident edges across different components is minimized.
Good partitions according to this model typically result in better balance across workers and reduced computation time in comparison to the trivial hash-based assignment of vertices to worker machines.
However, edge-based graph partitioning can still result in performance regressions~\cite{AKCV18,Sun18}.

To analyze the source of regressions, we performed a simple experiment of running a \texttt{Page Rank} algorithm implemented on top of Giraph utilizing various graph partitioning methods.
Figure~\ref{fig:hist} illustrates the histograms of running times for individual workers processing
a graph with $800M$ vertices and $80B$ edges.
As discussed above, partitions according to the \emph{vertex partitioning} model suffer from unequal distribution of edges across workers. A single overloaded partition can contain $1.92$x more edges than an average one, which results in $1.5$x longer execution time.
We also observe a high correlation ($\rho = 0.79$) between the number of edges assigned to a partition and the corresponding processing time in this experiment. Partitioning according to the \emph{edge partitioning} model yields a $1.08$x
running time improvement over the baseline, though there is still a noticeable imbalance between the fastest and the slowest worker machines. This can be explained by uneven distribution of vertices among workers.
\rlabel{line:R2.11}{Machines with more vertices have higher operational overhead such as serialization of sent messages whose number is proportional to the number of vertices on a worker.
Here we observe an $1.33$x imbalance in the number of vertices and a moderate correlation ($\rho = 0.62$) between the running time and the vertex count on the workers.}

In order to mitigate the issues described above we introduce a new strategy, \emph{vertex-edge partitioning}, which is designed to balance the number of vertices and edges across workers \emph{simultaneously}. As shown in Figure~\ref{fig:hist}, this is done at a cost 
of a lower edge locality (percentage of edges with both endpoints on the same machine), and thus, higher communication volume.
The resulting assignment results in a $1.17$x speedup over the hash-based model.
Motivated by the above experiment and a number of earlier studies~\cite{VLSG17,GHCIE17,AKCV18,Sun18},  we formalize a new model for graph partitioning which is suitable for real-world distributed graph processing systems.

\rlabel{line:R1.1}{
We now formally describe the model in the most general setting which allows one to require balance according to $d$ different unrelated weight functions. Let $G(V, E)$ be a graph with $d$ vertex weight functions $\w{1}, \dots, \w{d}: V \rightarrow \mathbb R^+$, each assigning a positive weight to every vertex in the graph. 
Let $\wj(V) = \sum_{v \in V} \wj(v)$ be the sum of weights of all vertices in the graph according to the $j$-th weight function.
Given an integer $k$ and a parameter $\eps > 0$, the goal is to find a partition of the vertex set $V$ into $k$ sets $V_1, \dots, V_k$ such that for each weight function $\wj$ and each part of the partition $V_i$ the sum of weights in $V_i$ is approximately the same and close to the average, i.e.  $\sum_{v \in V_i} \wj(v) = (1 \pm \eps) \frac{\wj(V)}{k}$. We call such partitions $\eps$-balanced.
Finally, among all such $\epsilon$-balanced partitions the goal is to find one that maximizes the number of edges whose both endpoints are contained within some part of the partition and hence minimizes the size of the cut.
This problem is referred to as \prob{Multi-Dimensional Balanced Graph Partitioning} (\prob{MDBGP}).

\rlabel{line:R3.4}{The simplest example of \prob{MDBGP} is the classic \textit{balanced graph partitioning problem} which is equivalent ot the vertex partitioning strategy described above and can be expressed using a single weight function $\w{1}(v) = 1$. Since $\w{1}(V) = |V|$ this requires that we maximize edge locality while ensuring that $|V_i| \approx \frac{|V|}{k}$.
Using two weight functions $\w{1}(v) = 1$ and $\w{2}(v) = deg(v)$ corresponds to requiring balance on the number of vertices and edges in the parts of the partition and hence corresponds to the vertex-edge partitioning approach described above.
 Indeed, $\w{2}(V) = 2|E|$ and hence in addition to balance on the number of vertices this requires that $\sum_{v \in V_i} deg(v) \approx  \frac{2|E|}{k}$.
However, the model is not restricted to vertex- and edge-balance (as in the aforementioned \emph{vertex-edge partitioning}) but can take arbitrary user-specified weights. In particular, when partitioning the vertices of the graph between the workers for load balancing, various weights modeling expected vertex activity can be used (historical data on individual vertex load, proxy values for the load such as PageRank, etc).}}

While a large body of work exists offering practical solutions for the one-dimensional version of the problem~\cite{KK95, DGRW12, UB13, TGRV14, ABM16, DKKOPS16, MLLS17, KKPPSAP17, BMSSS16}, as well as on theoretical foundations of graph partitioning~\cite{KNS09,AFKRS14,MM14}, literature on principled and scalable approaches for the multi-dimensional case is quite sparse~\cite{KK98,SKK99,OYL06,NU13}.
In particular, if the weight functions are unrelated to each other, one can easily construct examples when no feasible solution exists that satisfies all balance constraints even for two weight functions. 
However, it is empirically observed that instances coming from applications often allow balanced solutions for several weight functions of interest simultaneously.
For classical local search based algorithms such as~\cite{KL70} handling of multiple unrelated weight functions is challenging since imposing one balance constraint might violate another and hence finding a good local move becomes computationally hard.
We overcome this difficulty by using a continuous relaxation of the problem, which allows more flexibility for achieving balance in the search space. In order to obtain an integral solution, in the end we apply randomized rounding which preserves balance with high probability.

%

\subsection{Our Contributions}

We present a scalable algorithmic framework for the problem of balanced partitioning of large graphs according to multiple user-specified weight functions 
while maximizing the number of edges inside the resulting components.
Our framework consists of applying the \emph{projected gradient descent} on a standard relaxation with a suitably chosen projection method.
The relaxation is to maximize a non-convex quadratic function $f(\vx) = \frac12 \vx^T A \vx$ for $\vx \in \mathbb R^n$, where $A$ is the adjacency matrix, subject to a constraint $\vx \in K$ for a certain convex body $K$ defined by the weight functions.
Section~\ref{sec:algorithm} provides the exact description of the relaxation.
Note that the gradient descent step only uses a matrix-vector multiplication since $\nabla f = A \vx$, and thus, the algorithm allows a straightforward distributed implementation.

While applying projected gradient descent to solve non-convex optimization problems subject to convex constraints is a well-studied approach in non-linear optimization (Section 2.3,~\cite{B99}) and machine learning (Section 6.6,~\cite{JK17}), one has to overcome two technical challenges to make it applicable to the multi-dimensional graph partitioning problem: 1) projection step is computationally expensive,
2) existence of points with small gradient (saddle points) slows down convergence.

We show how to address the first challenge by designing efficient projection step algorithms tailored to the standard relaxation of \prob{MDBGP}.
While convergence to the projection point can be achieved using various alternating projections methods~\cite{D83}, for $d \le 2$ we give one-shot exact solutions with almost linear running time.

\begin{theorem}\label{thm:runtime}
Running time of the projected gradient descent step is $\Oh(|E| + |V| \log^{d-1} |V|)$ for $d \le 2$ and scales as $\Oh(|E|/m +|V| \log^{d-1} |V|)$ when distributed between $m$ machines.
\end{theorem}

In order to address the second challenge, we use small perturbations to get out of saddle points, where the perturbation vectors are sampled from a scaled $n$-dimensional Gaussian distribution.
We refer to the resulting algorithm as \algname, see Algorithm~\ref{alg:br}.
Convergence analysis of \algname remains an open problem. While noisy gradient descent is known to have fast convergence to a local optimum for non-convex optimization subject to equality constraints, if inequality constraints are allowed convergence analysis is unknown~\cite{GHJY15}. 
 
Our experimental results show that \algname scales to graphs with up to several billions of vertices and up to $10^{12}$ edges. 
We conducted an experimental evaluation of various graph partitioning strategies for
optimizing several real-world Giraph workloads. The results demonstrate that multi-dimensional balancing is 
a suitable objective for achieving performance improvements, providing speedups in the order of $10\%-30\%$ over the state-of-the-art one-dimensional partitioning strategies.
Compared to existing scalable graph partitioners, such as Social Hash Partitioner~\cite{KKPPSAP17}, Spinner~\cite{MLLS17}, and Balanced Label Propagation~\cite{UB13,MSS14},
the algorithm is conceptually simple and obtains close-to-perfect balanced partitions across multiple dimensions.


\subsection{Previous Work}

\rlabel{line:R2.5-1}{While one-dimensional balanced graph partitioning has been studied extensively and a number of tools exist~\cite{KK95, DGRW12, UB13, TGRV14, ABM16, DKKOPS16, MLLS17, KKPPSAP17} (see also surveys  by Bichot and Siarry~\cite{BS13} and by Bulu{\c{c}}~et~al.~\cite{BMSSS16}), to the best of our knowledge none of the practical algorithms for this problem have been previously based on running gradient descent on a continuous relaxation.}
Existing approaches are inherently discrete and are based on combinations of various discrete algorithms: greedy heuristics (METIS~\cite{KK95}, Fennel~\cite{TGRV14}), branch-and-bound~\cite{DGRW12}, label propagation and local search (balanced label propagation~\cite{UB13}, Social Hash Partitioner~\cite{KKPPSAP17}, Spinner~\cite{MLLS17}), as well as hybrid approaches (linear embedding method combined with various optimizations~\cite{ABM16}). Due to the combinatorial nature of these algorithms, their generalizations to the multi-dimensional case appear to be non-straightforward without substantial losses in performance, while our continuous relaxation handles multiple balance constraints uniformly. 
\rlabel{line:R2.5-2}{Compared to the one-dimensional version, existing literature on the multi-dimensional version is rather sparse~\cite{KK98,SKK99,OYL06,NU13} and the main publicly available tool for the problem is currently METIS~\cite{KK98,SKK99}.}

Vast literature exists on optimization of non-convex functions and the interest in this topic lately has been particularly high.
However, in the constrained case when the optimization has to be performed over a convex body, fairly little is known; see classic optimization literature~\cite{B99,WN99,BV04}. 
Recent results on the non-convex optimization problem subject to convex constraints and its special cases include~\cite{GHJY15,SQW15,AG16,GLM16,JK17}.
Closest to our work in terms of techniques is~\cite{LRSST10} who use projected gradient method to solve convex programs involving the max-norm and show how to solve large semidefinite programming relaxations of \prob{Max-Cut}. Their results are quite different from ours as we consider a balanced version of graph partitioning and expect our algorithms to be scalable; the largest instances handled by~\cite{LRSST10} have $|V| = 20K$ and $|E| = 40K$. Since we require that our algorithms scale to graphs with billions of edges, using existing general purpose software for constrained quadratic programming is also infeasible.




\section{Projected Gradient Descent}
\label{sec:algorithm}

For an integer $t$ we use notation $[t]$ to denote the set $\{1, \dots, t\}$.
The weighted $d$-dimensional balanced graph partitioning problem is defined by a collection of $d$ weight functions $\w 1, \dots, \w d$, where $\wj \colon V \to \mathbb R^+$.  
For a set $S \subseteq V$ we use notation $\wj(S) \equiv \sum_{v \in S} \wj_v$.



\begin{definition}[\prob{MDBGP}] 
Given a graph $G(V, E)$, an integer $k$ and a parameter $\eps > 0$, the \prob{Multi-Dimensional $\epsdef$-Balanced Graph $k$-Partitioning} problem is to find a partition of the vertex set $V$ into $k$ sets $V_1, \dots, V_k$ such that for each $j \in [d]$, it holds that  $\wj(V_i) = (1 \pm \eps)\frac{\wj(V)}{k}$ for all $i \in [k]$. Among all such partitions the goal is to find one that maximizes the number of edges whose both endpoints are contained within some part of the partition. 
\end{definition}

In this paper we focus on the $2$-partitioning problem; for the general variant of $k$-partitioning, we apply the algorithm recursively.
For $k = 2$ \prob{MDBGP} is equivalent to the following integer quadratic program:
\begin{tcolorbox}[
    standard jigsaw,
    opacityback=0,  
]
\begin{align*}
&\text{Maximize:  }  && \frac12  \sum_{(i_1, i_2) \in E} (x_{i_1} x_{i_2} + 1) && \\
&\text{Subject to:} && \left|\sum_{i = 1}^n \wj_i x_i\right| \le \eps \sum_{i = 1}^n \wj_i  && \forall j \in [d]\\
&&& x_i \in \{-1,1\} && \forall i \in V
\end{align*}
\end{tcolorbox}

\rlabel{line:R1.5}{
The interpretation of $x_i$ variables is that if $x_i = 1$ then $i \in V_1$ and if $x_i = -1$ then $i \in V_2$. The objective is then the same as in \prob{MDBGP} and counts the number of edges whose both endpoints are contained in some part of the partition. Indeed, an edge $(i_1,i_2)$ makes a contribution of $1$ to the objective when $x_{i_1} = x_{i_2}$ (and hence $x_{i_1} x_{i_2} = 1$) and $0$, otherwise (since $x_{i_1} x_{i_2} = -1$).
The constraints are equivalent to $-\epsilon \wj(V) \le \wj(V_1) -  \wj(V_2) \le \epsilon \wj(V)$. Adding or subtracting $\wj(V)$ to both sides and dividing by $2$ we have $\wj(V_i) = (1 \pm \eps) \frac{\wj(V)}{2}$ as required in \prob{MDBGP}.

After dropping the additive term the objective can be expressed as $f(\vx) =\frac12\vx^T A \vx$ and has gradient $\nabla f(\vx) = A \vx$ and Hessian $\nabla^2(f) = A$. 
Finally, we use a continuous relaxation of the above problem where we replace the integrality constraints with $x_i \in [-1,1]$ for all $i \in V$. A solution to this continuous relaxation can be converted into an integral solution using randomized rounding. Using independent random variables $X_i$ for each vertex such that $\Pr[X_i = 1] = \frac {1 + x_i} 2$ and $\Pr[X_i = -1] = \frac {1 - x_i} 2$ the expected value of the objective on the rounded solution $(X_1, \dots, X_{|V|})$ is the same as on the initial fractional solution $(x_1, \dots, x_{|V|})$ while all balance constraints are still approximately preserved with high probability by concentration bounds.
}


\subsection{Overview}\label{sec:overview}

\newcommand{\xt}[1]{\vx^{(#1)}}
\newcommand{\yt}[1]{\vy^{(#1)}}
\newcommand{\zt}[1]{\vz^{(#1)}}
\newcommand{\xti}[2]{x^{(#1)}_{#2}}
\newcommand{\yti}[2]{y^{(#1)}_{#2}}
\newcommand{\zti}[2]{z^{(#1)}_{#2}}
We propose the following algorithm for the multi-dimensional balanced graph partitioning problem based on the continuous relaxation described above. 
The algorithm is referred to as \texttt{Gradient Descent} (\algname), see Algorithm~\ref{alg:br}.
It computes a sequence of vectors $\set{\xt t}$, where $\xti ti \in [-1; 1]$ for all $i \in V$ and $t$.
Here $\xt 0$ is initialized with zero vector, and $\xt {t+1}$ is computed by applying projected gradient descent iteration to $\xt t$.
Each iteration consists of three steps.

\emph{Step 1: Adding noise.} We add Gaussian noise to $\xt t$ obtaining a noisy vector $\zt t$.
The noise is drawn from the $n$-dimensional Gaussian distribution $N_n(0,\eta_t)$ with zero mean and variance $\eta_t$ in each coordinate.
The addition of noise to $\xt t$ allows to escape from saddle points, e.g. $\xt 0 = 0$.

\emph{Step 2: Gradient descent.} We obtain $\yt t$ from the noisy vector $\zt t$ via a gradient descent step with step size $\gamma_t$. Note that the gradient at $\zt t$ is given as $A \zt t$ hence this step can be expressed as $\yt t = (I + \gamma_t A) \zt t$.

\emph{Step 3: Projection.} The resulting vector $\yt t$ is then projected on the feasible space $\mathcal B_\infty \cap \bigcap_{j = 1}^d \mathcal S^j_\eps$, where:
\begin{align*}
&\mathcal B_\infty = \set{\vx \in \mathbb R^n | \forall i \colon x_i \in [-1,1]} \\
&\mathcal S^j_\eps = \set{\vx \in \mathbb R^n  |\ |\sum_{i = 1}^n \wj_i x_i |\ \le\ \eps \sum_{i = 1}^n \wj_i} \text{ for } j \in [d],
\end{align*}
that is, $\mathcal B_\infty$ satisfies that $\|\mathbf x\|_\infty \le 1$ and $\mathcal S^j_\epsilon$ corresponds to the constraints imposed by the balance of weights according to the $j$-th weight function.

The final solution is obtained by rounding last $\xt t$: each vertex $i$ is assigned to part $V_1$ with probability $\frac {\xti ti + 1} 2$. Note that this ensures that the expected number of edges whose endpoints belong to the same part after this rounding is given as
$\frac12  \sum_{(i_1, i_2) \in E} (\xti t{i_1} \xti t{i_2} + 1)$.



\begin{algorithm}
	\SetKwInOut{Input}{input}
	\SetKwInOut{Output}{output}
	\SetKwRepeat{Do}{do}{while}
	\Input{Graph $G(V,E)$, $\eps \in [0,1]$, weight functions $w_1, \dots, w_d \colon V \to \mathbb R^+$}
	\nonl \textbf{parameters:} $I, \{\eta_t\}_{t = 0}^{I-1}, \{\gamma_t\}_{t = 0}^{I-1}$\\
	\Output{$\epsdef$-balanced partition w.r.t $\w 1, \dots, \w d$ of $V$ into $(V_1, V_2)$}
	\caption{\algname ($d$-Dimensional Balanced Graph $2$-Partitioning via Randomized Projected Gradient Descent)}
    \label{alg:br}
    $K = \mathcal B_\infty \cap \bigcap_{j = 1}^d \mathcal S^j_\eps$ \\
	$\xt{0} = \mathbf 0$; \\	
	\For{$t = 0$ \KwTo $I-1$}{
		$\zt{t} = \xt{t} + N_n(0,\eta_t)$;  \texttt{\hfill// Noise addition step\label{ln:gaussian-noise}}\\
		$\yt{t + 1} = (I + \gamma_t A) \zt{t} $; \texttt{\hfill// Gradient descent step\label{ln:grad-step}}\\
		$\xt{t + 1} = \argmin\limits_{\vx \in K} \|\yt{t + 1} - \vx\|_2$, 
		\texttt{\hfill// Projection step\label{ln:projection}}\\
	}
    $V_1 = V_2 = \emptyset$; \texttt{\hfill// Randomized rounding\label{ln:randomized-rounding}} \\
    \For{each $i \in V$} {
    	\texttt{~With probability $\frac {\xti ti + 1} 2$, let $V_1 = V_1 \cup \{i\};$ otherwise, $V_2 = V_2 \cup \{i\};$} \\
    }
\end{algorithm}

The algorithm uses parameters $\eta_t, \gamma_t$, and $I$, where $t$ is the iteration index. 
Here $\eta_t$ controls the magnitude of noise, $\gamma_t$ is the step size, and $I$ is the number of iterations.
We discuss the selection of parameters in the experimental Section~\ref{sec:experiments}.

\subsection{Projection}
\label{SEC:PROJECTION}

In the projection step of \algname (Line~\ref{ln:projection}) we need to find $\argmin_{\vx \in K} \|\vy^{(t + 1)} - \vx\|_2$, where $K = \mathcal B_\infty \cap \bigcap_{j = 1}^d \mathcal S_\eps^j$. Denoting $\vy^{(t + 1)}$ as $\vy$ we formulate this step as an optimization problem:
\begin{tcolorbox}[
    standard jigsaw,
    opacityback=0,  
]
\begin{align*}
& \text{Minimize:} && f(\vx) = \|\vx - \vy\|_2^2&& \\ 
& \text{Subject to:} && g_i = x_i^2 - 1 \le 0 && \forall i \in [n]\\
&&& \hj_+ = \sum_{i = 1}^n \wj_i x_i - \eps \le 0 && \forall j \in [d]\\
&&& \hj_- = -\sum_{i = 1}^n \wj_i x_i - \eps \le 0 && \forall j \in [d]
\end{align*}
\end{tcolorbox}

The optimum solution to the optimization problem has to satisfy KKT conditions:

\begin{tcolorbox}[
    standard jigsaw,
    opacityback=0,  
]
Stationarity:
\begin{align*}
\vy - \vx = \sum_{i = 1}^n \mu_i x_i \ve_i + \sum_{j = 1}^d (\muj_+ - \muj_-)\sum_{i = 1}^n \wj_i \ve_i
\end{align*}
Complementary slackness 1:
\begin{align*}
&\mu_i (x_i^2 - 1) = 0, && \forall i \in [n] 
\end{align*}
Complementary slackness 2:
\begin{align*}
&\muj_+ \left(\sum_{i = 1}^n \wj_i x_i - \eps\right) = 0,   &&\forall j \in [d] \\
&\muj_- \left(\sum_{i = 1}^n \wj_i x_i + \eps\right) = 0,  && \forall j \in [d] 
\end{align*}
\end{tcolorbox}

Here $\mu_i, \muj_+, \muj_i \ge 0$ are the dual variables and $\ve_i$ is the $i$-th standard unit vector.
It is a standard fact (see~\cite{BV04}, Chapter 5.5.3) that for convex optimization subject to linear constraints Stationarity, Complementary slackness and Primal/Dual feasibility are necessary and sufficient conditions for the optimum solution. Thus we just focus on satisfying these conditions below.

Let $\gamma_i = \sum_{j =1}^d (\muj_+ - \muj_-) \wj_i$. Then by Stationarity for each $i$ we have $y_i - x_i  = \mu_i x_i + \gamma_i$. Consider the following three cases:

\textbf{Case 1.} $(y_i > 1 + \gamma_i)$. If $\mu_i = 0$ then by Stationarity $x_i = y_i - \gamma_i > 1$ which violates primal feasibility conditions. Therefore $\mu_i > 0$ and $x_i^2 = 1$ by Complementary slackness 1. Among the two roots $x_i = 1$ and $x_i = -1$ the second root can be ruled out and hence $x_i = 1$.
Indeed, if $x_i = -1$ then by Stationarity $y_i + 1 = -\mu_i + \gamma_i$ which contradicts $\mu_i > 0$ and $y_i > 1 + \gamma_i$.

\textbf{Case 2.} $(y_i < - 1 + \gamma_i)$. This case is symmetric to the previous one and thus $x_i = -1$ in this case.

\textbf{Case 3.} $y_i \in [-1 + \gamma_i, 1 + \gamma_i]$. First we show that $\mu_i = 0$. Indeed, assume that $\mu_i > 0$. Then $x_i = \pm 1$ by Complementary slackness 1.
Both cases lead to contradiction:
\begin{enumerate}
\item $(x_i = 1)$. By Stationarity $y_i - 1 = \mu_i + \gamma_i$ which contradicts with $y_i \le 1 + \gamma_i$ and $\mu_i > 0$.
\item $(x_i = -1)$. Similarly to the above by Stationarity we have $y_i + 1 = - \mu_i + \gamma_i$ which is a contradiction with $y_i \ge -1 + \gamma_i$ and $\mu_i > 0$.
\end{enumerate}
Therefore in this case we have $\mu_i = 0$ and hence by Stationarity $x_i = y_i - \gamma_i$.

Let $\lambda_j = \muj_+ - \muj_-$ and assume that these values are known to the algorithm. For $z \in \mathbb R$ we use notation $[z] = \min(1, \max(-1, z))$ for the truncated linear function.
Using the analysis above the projection step is simply $x_i = [y_i - \sum_{j = 1}^d \lambda_j \wj_i]$. It remains to show how to find $\set{\lambda_j}$.

Note that from Complementary slackness 2 it follows that either $\muj_+ = 0$ or $\muj_- = 0$ since both of these values being positive leads to a contradiction. 
This leads to three cases: 1) $\muj_+ = 0, \muj_- > 0$, 2) $\muj_- = 0, \muj_+ > 0$ and 3) $\muj_+ = \muj_- = 0$ which correspond to the three possibilities for $sign(\lambda_j)$.
For each of the $d$ dimensions we can try all three choices.
For a fixed guess of the signs let $S_+ = \{j \colon \lambda_j > 0\}$, $S_0 = \{j \colon \lambda_j = 0\}$ and $S_- = \{j \colon \lambda_j < 0\}$.
Assuming a correct guess of $sign(\lambda_j)$ for each of the dimensions the optimization problem above reduces to the following:

%

\begin{proposition}\label{prop:approx-to-exact-reduction}
For the correct guess of $sign(\lambda_j)$ for all $j \in [d]$ it suffices to find the optimum of the above optimization problem without the constraints for $j \in S_0$. This optimum is unique.
\end{proposition}

The proof is given \iffull Appendix~\ref{app:proofs}\else in the full version\fi.
Using Proposition~\ref{prop:approx-to-exact-reduction} and trying all guesses for $sign(\lambda_j)$ we can reduce the projection step to $3^d$ instances of the following optimization problem:
\begin{tcolorbox}[
    standard jigsaw,
    opacityback=0,  
]
\begin{align*}
&\text{Minimize:} && f(\vx) = \|\vy - \vx\|_2^2&& \\ 
&\text{Subject to:} && g_i = x_i^2 - 1 \le 0 && \forall i \in [n]\\
&\sum_{i = 1}^n \wj_i x_i = \eps, & \forall j \in S_+; & \sum_{i = 1}^n \wj_i x_i = -\eps, && \forall j \in S_-
\end{align*}
\end{tcolorbox}
which can be done by finding numbers $\lambda_j > 0$ for $j \in S_+$ and $\lambda_j < 0$ for $j \in S_-$ and setting $x_i = [y_i - \sum_{j \in S_+ \cup S_-} \lambda_j w_{ij}]$. The choice of $\lambda_j$'s has to satisfy the constraints $\sum_{i = 1}^n w_{ij} x_i = \eps$ for all $j \in S_+$ and $\sum_{i = 1}^n w_{ij} x_i = -\eps$ for all $j \in  S_-$. In the analysis below we assume that $d = |S_+ \cup S_-|$ corresponds to the ``effective dimension'' of the problem.

\subsection{Exact Projection Algorithms}
\label{sec:exact-projection}

\paragraph*{Projection for $d = 1$}
As a warm up, we first show how to perform exact projection for $d = 1$ in $\Oh(n \log n)$ time, proving Theorem~\ref{thm:runtime} for $d=1$.
This can be further improved to $\Oh(n)$ using a more careful approach~\cite{MSMJ03}. However, to the best of our knowledge, no fast algorithm is known for $d > 1$ which is the main focus of our work. 
%
Dropping the second index to simplify presentation (that is, $w_i = \w 1_i$) and using the fact that $x_i = [y_i - \lambda w_i]$ we have:

\begin{align*}
\sum_i w_i x_i = &\sum_{i \colon y_i \ge 1 + \lambda w_i} w_i \ + \sum_{i \colon y_i \le -1 + \lambda w_i} -w_i \ +\\
&+ \sum_{i \colon y_i \in (-1+\lambda w_i, 1+\lambda w_i)} w_i (y_i - \lambda w_i).
\end{align*}

We introduce notation $h_i(\lambda)$ where each $h_i$ is the following piecewise linear function:
\begin{align*}
h_i(\lambda) = \begin{cases}
w_i  & \text{ if } \lambda < (y_i - 1) / w_i \\ 
w_i (y_i - \lambda w_i) & \text{ if } \lambda \in [(y_i - 1) / w_i, (y_i + 1) / w_i]\\
-w_i & \text{ if } \lambda > (y_i + 1) / w_i \\
\end{cases}
\end{align*}
Thus $\sum_{i = 1}^n w_i x_i = \sum_{i = 1}^n h_i(\lambda)$ and the problem reduces to finding $\lambda^*$ such that $\sum_i h_i(\lambda^*) = \pm \eps$ where the sign depends on whether our dimension is in $S_+$ or $S_-$.
Since $w_i \ge 0$ for all $i$ each $h_i$ is monotone in $\lambda$ and so the function $\sum_i h_i$ is a monotone piecewise linear function. The value of $\lambda^*$ can be found in $\Oh(\log n)$ iterations of binary search where each iteration requires $\Oh(n)$ time to evaluate the sum. This gives the overall running time of $\Oh(n \log n)$.
See Figure~\ref{fig:1D-projection} for an illustration.

\begin{figure}[!tb]
	\centering
\begin{tikzpicture}[scale=1.5]
	\pgfmathsetmacro{\left}{-2}
	\pgfmathsetmacro{\right}{2}
	
	\draw (-1,-1) -- (1, -1) -- (1, 1) -- (-1, 1) -- cycle;	
	\draw (\left, 0.95) -- (\right, -1.15) node [below] {$\langle w, \vx \rangle = -\eps$};
	\coordinate (wend) at (0.225, 0.5);
	\coordinate (uend) at (\right, -0.95);
	\draw[name path = UL] (\left, 1.15) -- (uend) node [above=0.3] {$\langle w, \vx \rangle = \eps$};
	\draw[->,name path = W] (0,0) -- (wend) node[right] {\small $w$};
	
	\coordinate (y) at (-1.5,1.5);
	\path[name path=left border] (-1,-1) -- (-1, 1);
	\draw[name intersections={of=left border and UL, by={x}}];
	\draw[fill=white] (y) circle (0.3mm) node[right] {$\vy$};
	\draw[fill=white] (x) circle (0.3mm) node[above right] {$\vx$};
	
	\path[name path = xcoord] (x) -- ($ (x) - (1,0)$);
	\path[name path = yparallel] (y) -- ($ (y) - 3*(wend) $);
	\draw[name intersections={of=xcoord and yparallel, by={x'}}];
	\draw[fill=white] (x') circle (0.3mm) node[below] {$\vy - \lambda^* w$};
	\draw[dotted,->,thick] (y) -- (x');
	\draw[dotted,->,thick] (x') -- (x);
	
	\draw[name intersections={of=UL and W, by={near0}}];
	\tkzMarkRightAngle[size=.12](uend,near0,wend);
	\draw[name intersections={of=UL and yparallel, by={yproj}}];
	\tkzMarkRightAngle[size=.12](y,yproj,uend);
	\coordinate (leftbottom) at (-1,-1);
	\tkzMarkRightAngle[size=.1](x',x,leftbottom);

\end{tikzpicture}
	\caption{One-dimensional projection. First, the initial point $\vy$ is moved by vector $-\lambda^* w$, which is an orthogonal vector to planes, corresponding to balance constraints. Then the resulting point is projected on the cube.}
	\label{fig:1D-projection}
\end{figure}

\paragraph*{Projection for $d = 2$}
For $d = 2$ we need to find $(\lambda_1, \lambda_2)$ such that $\sum_{i = 1}^n \hj_i(\lambda_1, \lambda_2) = \pm \eps$ for $j = 1,2$, where $\hj_i(\lambda_1, \lambda_2)$ is defined below. 
\begin{align*}
\hj_i(\lambda_1, \lambda_2) = \begin{cases}
\wj_i & \text{ if }  \sigma_i < y_i - 1 \\ 
\wj_i (y_i - \sigma_i) & \text{ if } \sigma_i \in [y_i - 1, y_i + 1]\\
-\wj_i & \text{ if } \sigma_i > y_i + 1 \\
\end{cases}
\end{align*}
where $\sigma_i =\lambda_1 \w 1_i + \lambda_2 \w 2_i$.
The projection process is shown in Figure~\ref{fig:2D-projection}.
In Appendix~\ref{app:projection2d}, we prove Theorem~\ref{thm:runtime} for $d=2$ showing that nested binary search can be used to solve this problem in $\Oh(n \log n)$ time.

\begin{figure}[!tb]
	\centering
\begin{tikzpicture}[scale=1.2]
	\pgfmathsetmacro{\left}{-2}
	\pgfmathsetmacro{\right}{2}
	
	\draw (-1,-1) -- (0.6, -1) -- (0.6, 0.6) -- (-1, 0.6) -- cycle;	
	
	\coordinate (y) at (-3,1.5);
	\coordinate (x) at (-1, -0.33);
	\draw[fill=white] (y) circle (0.3mm) node[right] {$\vy$};
	\draw[fill=white] (x) circle (0.3mm) node[above right] {$\vx$};
	
	\coordinate (x') at ($ (x) - (2.5,0)$);
	\draw[fill=white] (x') circle (0.3mm) node[below] {$\vy - \lambda_1^* \w 1 - \lambda_2^* \w 2$};
	\draw[dotted,->,thick] (y) -- (x');
	\draw[dotted,->,thick] (x') -- (x);
	
	\coordinate (y1) at ($ (y) - (0.6,0.8) $);
	\draw[fill=white] (y1) circle (0.3mm) node[left] {$\vy - \lambda_1 \w 1$};
	\coordinate (y2) at ($ (y) + (x') - (y1) $);
	\draw[fill=white] (y2) circle (0.3mm) node[right] {$\vy - \lambda_2 \w 2$};
	\draw[dotted,->] (y) -- (y1);
	\draw[dotted,->] (y) -- (y2);
	\draw[dotted,->] (y1) -- (x');
	\draw[dotted,->] (y2) -- (x');
\end{tikzpicture}
	\caption{Two-dimensional projection. Initial point $y$ is moved by vector $\lambda_1 \w 1 + \lambda_2 \w 2$ and then projected on the cube.}
	\label{fig:2D-projection}
\end{figure}

\newpage

\section{Implementation}
\label{sec:implementation}

\subsection{Projection algorithms}
\label{sec:impl_projection}

We considered the following three methods for the projection step (Algorithm~\ref{alg:br}, Line~\ref{ln:projection}). Their theoretical properties are summarized in Table~\ref{table:proj-properties}.

\begin{compactitem}
  \item \emph{Alternating projections:} A standard approach for projection on the intersection of convex sets is the alternating projections method (see~\cite{BD03}). 
  It is easy to implement projections on $\mathcal B_\infty$ and $\cap_{j = 1}^d S^j_\eps$ separately.
	Since both are convex bodies by alternating projections on each of them one can guarantee convergence to a point in the intersection, but there is no guarantee that this point will be the actual projection. In practice, we are able to achieve slightly better balance by modifying this approach slightly and projecting on $S^j_0$ instead of $S^j_\eps$. This still ensures that we get a point in the intersection in the end.
  \item \emph{Dykstra's projection:~\cite{Dykstra}} We also considered Dykstra's projection algorithm~\cite{D83}.
        This is a modification of the alternating projections method which is guaranteed to converge to the projection.
  \item \emph{Exact projection for $d \le 2$:} This is the algorithm presented in Section~\ref{SEC:PROJECTION}. In our experiments Dykstra's algorithm and exact projection give similar results, since they find approximately the same projection point. 
\end{compactitem}

\begin{table}

\begin{center}
\begin{tabular}{|c|c|c|c|}
     \hline
     & $d$& Output & Time required  \\
     \hline
     Alternating & any & $\vx \in K$ & Until convergence  \\
     \hline
     Dykstra's & any & projection& Until convergence \\
     \hline
     Exact (ours) & $d \le 2$ & projection& $\Oh(n \log^{d - 1} n)$\\
     \hline
\end{tabular}
\end{center}
\caption{Theoretical properties of projection methods.} \label{table:proj-properties}
\end{table}

In Section~\ref{sec:params} we study how quality of partitions produced by \algname depends on choice of one of the projection methods above.
Since the exact projection algorithm is computationally the most expensive, in our experiments we mostly use the alternating projections method.
Moreover, since in practice each iteration of alternating projection is computationally expensive, in the intermediate iterations we project on each plane and the cube only once, while in the last iterations we run the alternating projections method until convergence.
We refer to this choice as ``one-shot'' alternating projection below.


\subsection{Adaptive Step Size}
\label{sec:adapt}

Recall that Algorithm~\ref{alg:br} has the following parameters: Gaussian noise variances for each step $\set{\eta_t}$ and step size parameters $\set{\gamma_t}$.
Due to the spectral properties of the adjacency matrix in our experiments the algorithm doesn't encounter any saddle points other than the initial point $\vx = 0$. Therefore it suffices to add Gaussian noise only at first iteration (that is, $\eta_t = 0$ for $t \neq 0$).

The simplest choice of the step size parameters $\set{\gamma_t}$ is constant, but it gives suboptimal results in our experiments.
Carefully chosen step size parameters for different iterations not only gives better performance but can also be used to ensure that convergence can be reached in a fixed number of steps.
In section~\ref{sec:params} we discuss how to choose the step size to achieve good performance on a wide range of graphs.

The choice of step size parameters is complicated by the projection step. The change in the objective function and the progress towards an integral solution can both be related to the progress in Euclidean distance $\|\vx_t - \vx_{t + 1}\|$ between the iterations.
While consistent progress in Euclidean distance can be ensured by multiplying the gradient by an appropriate amount after the projection the actual progress can be much smaller. 

Another important implementation detail is our handling of vertices which are close to integral.
When the number of such vertices becomes large the progress of the algorithm can slow down.
This is due to the fact that while the gradient vector is still large all of its large components correspond to already integral vertices and point to the outside of the feasible region.
These large components can then dominate in the computation of the projection step which leads to slow convergence. In order to avoid this issue we ``fix'' such vertices so that they become integral and no longer participate in the gradient update and the projection step. As we show in Section~\ref{sec:params} this results in noticeable improvements in the quality of the resulting partitions.



\subsection{Partitioning Into k Buckets}
\label{sec:k-partitioning}

For partitioning into more than two buckets two main approaches are typically considered. We use the second approach due to its higher efficiency.

\emph{Problem relaxation for $k$ buckets:} For each vertex $i$ and bucket $j$ we can introduce a variable $p_{ij}$ corresponding to whether $i$ belongs to bucket $j$ and then adjust the relaxation accordingly.
    Our algorithm \algname can then be modified to handle such relaxation.
    The main drawback of this approach is that it requires $\Oh(k \cdot |E|)$ communication per iteration, which makes it infeasible for partitioning large graphs into many buckets.
    
 \emph{Recursive partitioning:} The graph is partitioned recursively $\lceil \log_2 k \rceil$ times into two parts.
    While there are cases when recursive partitioning can result in a suboptimal partition regardless of the underlying algorithm, this approach requires $\Oh(|E|)$ memory, $\Oh(|E|)$ operations per iteration and $\Oh(\log k)$ runs of \algname, which makes it applicable to very large graphs. For simplicity we only show results for $k$ being powers of two but the algorithm can be modified to handle any $k$ by changing the coefficients in the balance constraints.

\section{Experiments}\label{sec:experiments}

We design our experiments to understand how well the new partitioning algorithm, \algname, behaves on real-world datasets 
and how it affects the performance of distributed graph processing. 
As pointed out in Section~\ref{sec:intro}, we are not
aware of an alternative scalable approach for solving the {\it multi-dimensional} balanced partitioning.
However, some of the existing techniques for one-dimensional partitioning
can be adapted for the multi-dimensional case. Next we discuss several such techniques, which
are evaluated together with the newly proposed algorithm.

\texttt{Hash} is the simplest partitioning strategy that assigns vertices to worker machines
  by hashing the vertex identifiers. Hashing is stateless, extremely fast in practice, and requires no 
  preprocessing of the graph, which made it the default strategy in Giraph. The main disadvantage is
  that the majority of sent messages are
  non-local and may results in significant communication.
  
\texttt{Spinner} is a graph partitioning algorithm that can be applied to process large-scale 
  graphs in a distributed environment~\cite{MLLS17}.
  The algorithm is based on the label propagation technique in which vertices exchange their
  labels trying to pick the most frequent label among its neighbors. This process guarantees a high
  number of adjacent vertices having the same label, which are then assigned to the same worker.
  \texttt{Spinner} does not enforce a strict balance across partitions but integrates score
  functions that penalize imbalanced solutions.
  
\texttt{BLP} is another approach based on the balanced label propagation based on combining the ideas of Ugander and Backstrom~\cite{UB13} and Meyerhenke et al.~\cite{MSS14}.
  On the first step, the method creates a size-constrained clustering of the input graph using significantly more
  clusters than the number of available machines, $k$. In our implementation, we construct $c \times k$ clusters
  for $c = 1024$ and forbid a cluster to contain more than $\frac{|V|}{c \times k}$ vertices and
  $\frac{|E|}{c \times k}$ edges. On the second step, we randomly merge the clusters into $k$ partitions, 
  which results in the multi-dimensional balance even if the original clusters have different sizes.
  
\texttt{SHP} is a distributed graph partitioner~\cite{SH16,KKPPSAP17} that is based on a classical local
  search heuristic~\cite{KL70}. Although \texttt{SHP} does not provide balancing on multiple dimensions,
  it supports a mode with several dimensions whose final balance is not guaranteed.
  The algorithm works by balancing on a new dimension, which is a combination of the specified dimensions.
  We configure \texttt{SHP}
  to find solutions having the same number of edges (with a higher coefficient in the combination) and the same
  number of vertices (with a lower coefficient) in every partition.

We implemented the algorithms and
extensively experimented with the Giraph framework, which is used as the primary tool for
large-scale graph analytics at Facebook~\cite{Chi11,giraph}. Although the evaluation is
performed with the single distributed graph processing system, we believe that our main
conclusions are valid for other frameworks relying on the vertex-centric programming model.
For our experiments, we use four large social networks that are publicly available~\cite{snapnets}.
\texttt{LiveJournal}, \texttt{Orkut}, \texttt{Twitter}, and \texttt{Friendster} are undirected graphs
containing $4.8$, $3.1$, $41$, and $65$ million of vertices and $0.04$, $0.12$, $1.2$, and $1.8$ billion of edges, respectively.
In addition, we experiment with several large subgraphs of the Facebook friendship graph
that serve to demonstrate scalability of our approach and its performance on real-world data.
We denote the graphs by \texttt{FB-X}, where \texttt{X} indicates the (approximate) number of billions of edges; 
this data is anonymized before processing.

Next we analyze the quality of the solutions produced by the algorithms on our
dataset (Section~\ref{sec:multi}) and evaluate various 
graph partitioning strategies for speeding up distributed graph processing for real-world workloads (Section~\ref{sec:digraph}).
Section~\ref{sec:params} investigates various parameters of \algname.

\subsection{Multi-Dimensional Partitioning}
\label{sec:multi}

Our initial experiments (see Figure~\ref{fig:hist}) and earlier works~\cite{GLGBG12,LGHC17,MLLS17} indicate that two important
dimensions for the performance of Giraph jobs are the number of vertices and the number of edges. For this reason,
we specify two weights for the vertices, $\w 1_v = 1$ and $\w 2_v = \deg(v)$ for all $v \in V$.
Recall that our primary goal is to guarantee almost perfect balance for the two dimensions, 
as even a single overloaded partition affects the job performance. Figure~\ref{fig:imbalance_small} illustrates the resulting
vertex and edge imbalance of the solutions on the public networks for three algorithms, \texttt{Spinner}, \texttt{BLP}, and
\texttt{SHP}, using $k=2$ and $k=8$ partitions. The imbalance is defined as $\left(\frac{\max_i w(V_i)}{\avg_i w(V_i)} - 1\right)$, where
the maximum and the average are taken over the total weight of all $k$ constructed partitions.
We do not include the results for \texttt{Hash} and \algname, as the corresponding values are below
$0.01$ for the instances.

\begin{figure}[t]
  \centering
  \includegraphics[width=0.9\columnwidth]{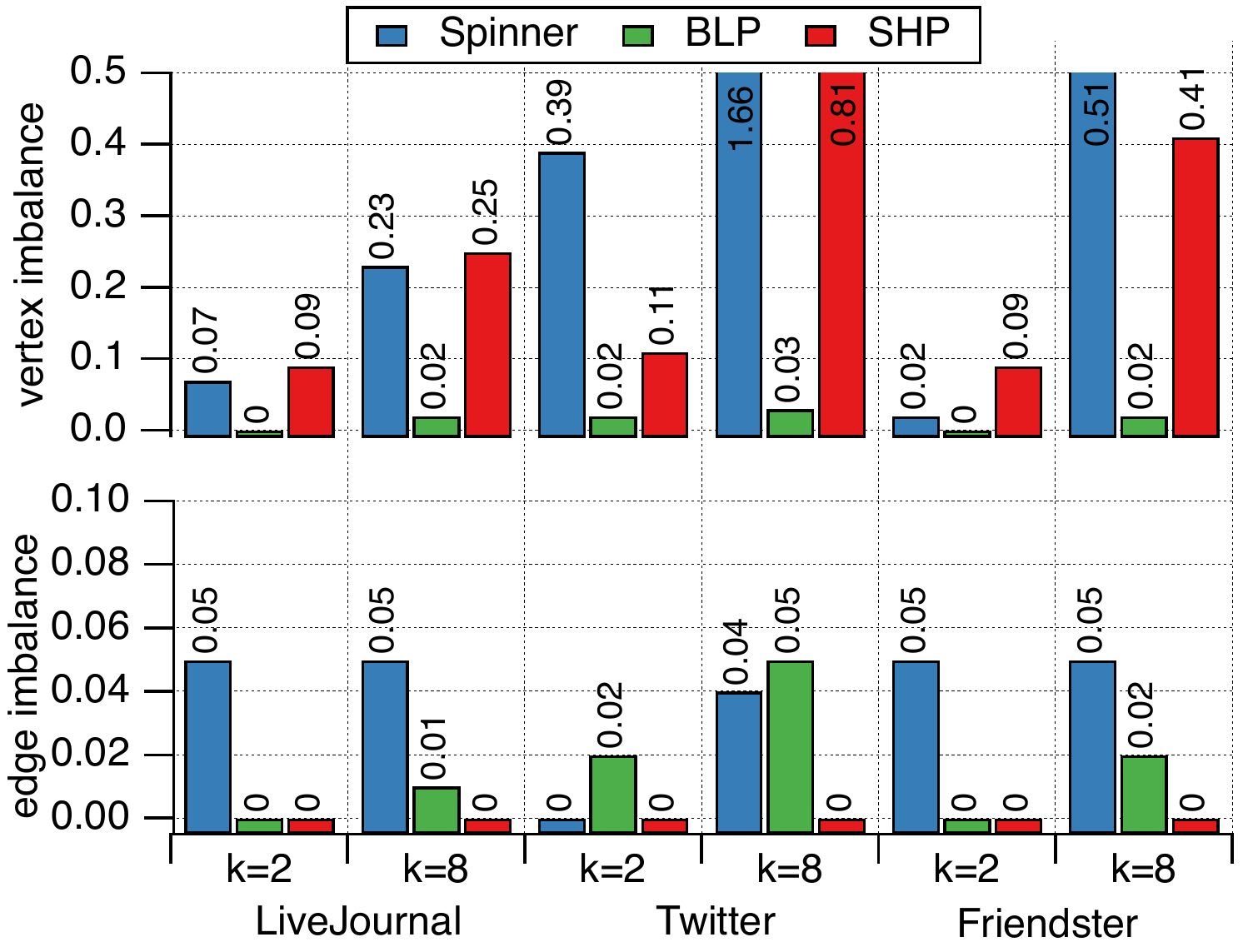}
  \caption{Vertex and edge imbalance ($\frac{\max_i w(V_i)}{\avg_i w(V_i)} - 1$) of the solutions created by
    different algorithms on the three public networks with $k \in \{2, 8\}$. 
    Lower values correspond to more balanced partitions.
    \texttt{Hash} and \algname yield near-balanced solutions for the instances.} 
  \label{fig:imbalance_small}
\end{figure}

We observe that two algorithms, \texttt{Spinner} and \texttt{SHP}, are not suitable for the multi-dimensional
variant of the problem. For dense graphs with a highly skewed degree distribution (as in \texttt{Twitter}), 
the algorithms cannot simultaneously provide balance on the two dimensions. With the default setting, these
two algorithms generate solutions in which some of the partitions contain $1.5-2$x more vertices than the
average one. We tried to modify the techniques by adjusting relative weights of their penalty functions
for vertex and degree counts in resulting partitions. However, we were not able to design universal
penalty weights that work for all instances. A similar behavior regarding the resulting balance is observed for our internal graphs,
\texttt{FB-3B}, \texttt{FB-80B}, and \texttt{FB-400B}.
In contrast, \texttt{Hash}, \algname, and \texttt{BLP}
produced nearly-balanced (that is, having $\eps \le 0.05$ both for vertex and edge counts) solutions
for all the instances. 
With this in mind, we exclude \texttt{Spinner} and \texttt{SHP} from further experiments.

Next we compare the quality of our algorithm as measured by the resulting edge locality, that is,
the percentage of uncut edges with both endpoints in the same partition. The metric represents the
fraction of local messages in Giraph jobs and corresponds to a possible reduction in communication 
between the worker machines. Figure~\ref{fig:quality_small} reports
the results of \texttt{Hash}, \algname, and \texttt{BLP} on the public dataset.
Unsurprisingly, \algname and \texttt{BLP} outperform the \texttt{Hash} algorithm in the experiment, 
as the latter keeps only $\frac{1}{k}$ of all the edges in the same partition.
The resulting edge locality of \algname and \texttt{BLP} are close for the three graphs, though
\algname typically achieves a higher locality by $2\%-5\%$.

\begin{figure}[t]
  \centering
  \includegraphics[width=0.9\columnwidth]{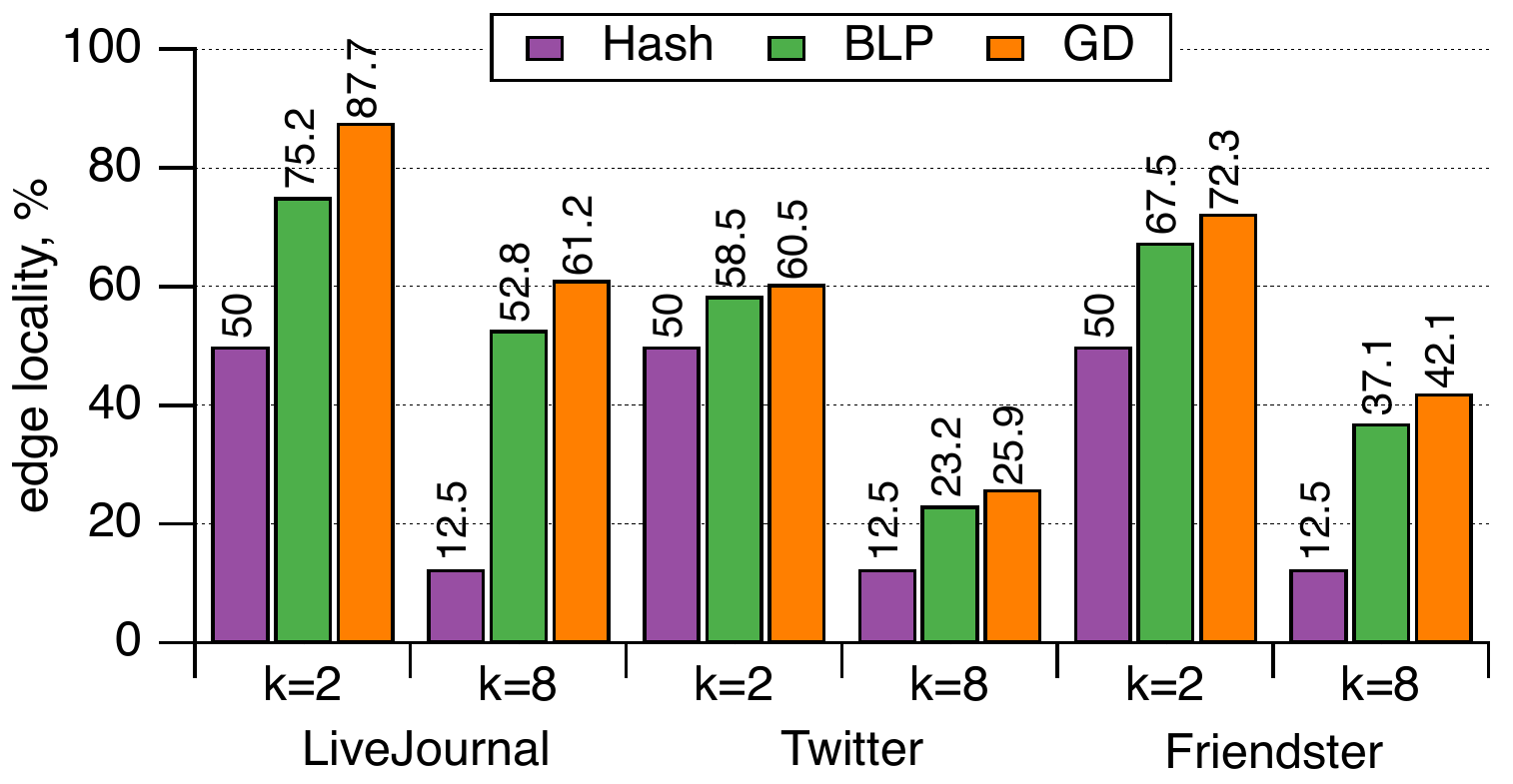}
  \caption{The percentage of local (uncut) edges produced by the three algorithms
    for the public graphs with $k \in \{2, 8\}$.
    Higher values indicate better solutions.
    \algname achieves higher locality in all cases.} 
  \label{fig:quality_small}
\end{figure}

Figure~\ref{fig:quality_large} shows the experiments on the Facebook friendship graphs.
Here we use a larger number of partitions, $k$, which more accurately represent the real-world
Giraph use case. Again, \texttt{Hash} produces solutions having the lowest edge localities.
In fact, over $99\%$ of the edges are cut using the partitioning strategy for an instance with
a hundred partitions. This is in agreement with our measurements of the typical percentage of cross-worker
Giraph messages in the production environment.
\rlabel{line:R2.7-reviewer}{
On the other hand, we observe a bigger advantage of \algname over \texttt{BLP}; the locality difference
is around $10\%-20\%$ for $k=16$ and $5\%-10\%$ for $k=128$.}
The balanced label propagation algorithm, \texttt{BLP}, could be configured to produce better results by
decreasing its cluster size threshold, $c$. However, this results in an imbalanced solution with $\eps > 0.05$
for the largest instance with $k=128$. Hence, we keep the value of $c=1024$ for all the experiments.
\begin{figure}[t]
  \centering
  \includegraphics[width=0.9\columnwidth]{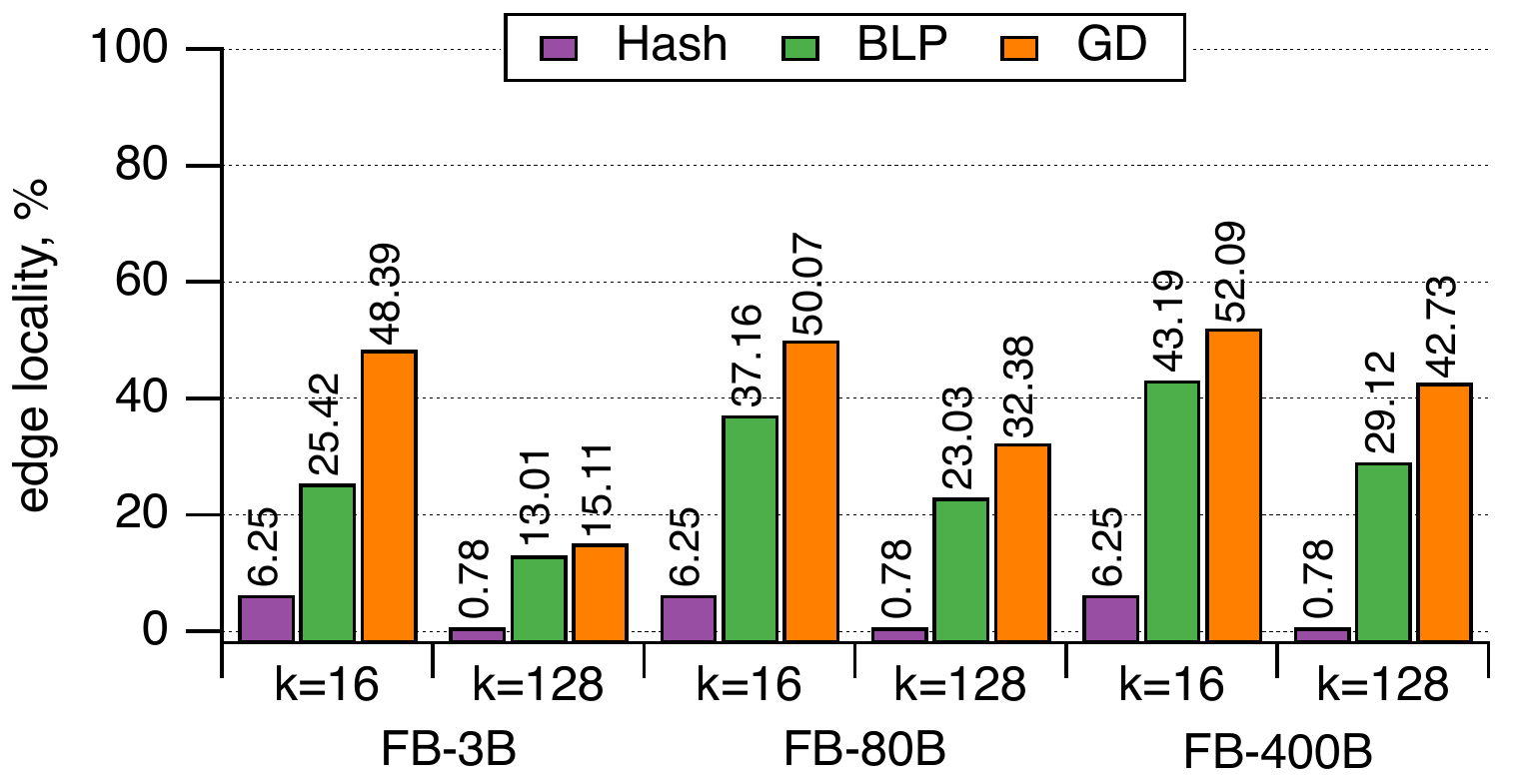}
  \caption{The percentage of local (uncut) edges produced by the three algorithms
    for various subgraphs of the Facebook friendship graph with $k \in \{16, 128\}$.
    Higher values indicate better solutions.
    \algname achieves higher locality in all cases.}
  \label{fig:quality_large}
\end{figure}
\rlabel{line:R2.7}{The main difference between FB graphs and publicly available graphs is the number of edges. The main reason why on FB graphs \algname performs better compared to other algorithms is poor performance of existing local-search based methods on large graphs in the multi-dimensional case. This is most obvious in Figure~\ref{fig:quality_large} for $k = 128$ as one can see that GD is gaining a larger advantage over BLP as the size of the graph grows (3B $\rightarrow$ 80B $\rightarrow$ 400B).}

Overall we conclude that \algname generates solutions of higher quality than \texttt{BLP} and \texttt{Hash}
on all examined instances. Therefore, we utilize the algorithm to experiment with distributed graph processing in the next section.
We present results for $3$- and $4$-dimensional experiments in \iffull Appendix~\ref{app:more_experiments}\else the full version\fi.

\subsection{Distributed Graph Processing}
\label{sec:digraph}

In this section we conduct an experimental evaluation of various graph partitioning strategies for
speeding up distributed graph processing. Here we argue and experimentally demonstrate that multi-dimensional balancing is 
a suitable objective for the application.
We experiment with four graph algorithms implemented in Giraph.
\texttt{Page Rank} and \texttt{Connected Components}, are popular benchmarks for verifying the performance
of distributed systems. \texttt{Page Rank} iteratively propagates vertex ranks through adjacent edges;
our implementation performs $30$ iterations for the algorithm. For the \texttt{Connected Components} algorithm,
we use a simple label propagation technique in which vertices iteratively update their labels based
on the minimum label of their neighbors; for our graphs, the process converges after at most $50$ rounds.
The other two algorithms, \texttt{Hypergraph Clustering} and \texttt{Mutual Friends}, 
are production applications for large-scale graph analytics at Facebook. The former is used to 
find a certain clustering of the input graph by converting it to a hypergraph. The latter builds a
set of features for friend recommendation on Facebook. Both applications extensively exchange messages
between adjacent vertices, which adds a significant communication overhead.

\begin{figure}[t]
  \centering
  \includegraphics[width=0.99\columnwidth]{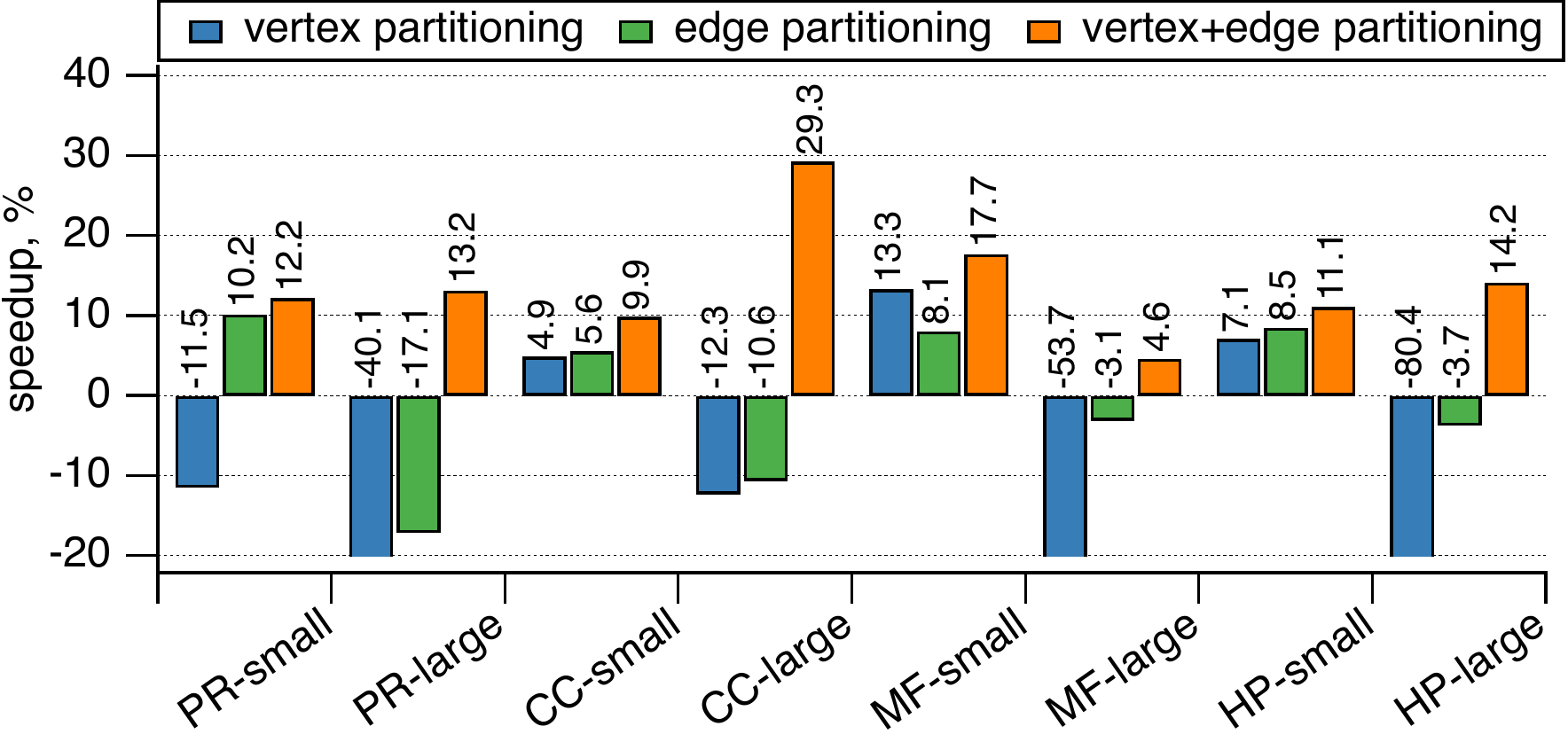}
  \caption{Speedup of Giraph jobs using various partitioning strategies relative to \texttt{Hash} measured for 
    \texttt{Page Rank} (\texttt{PR}), \texttt{Connected Components} (\texttt{CC}), 
    \texttt{Hypergraph Clustering} (\texttt{HC}), and \texttt{Mutual Friends} (\texttt{MF}), which are applied on
    \texttt{FB-80B} (\emph{small}) and \texttt{FB-400B} (\emph{large}) graphs.
    Positive values indicate improvements, negative ones indicate regressions.
    Vertex+edge partitioning always results in performance improvement.} 
  \label{fig:digraph}
\end{figure}

Figure~\ref{fig:digraph} depicts the results of our experiment.
Since we are interested in the impact of various partitioning policies on the performance of Giraph, we report the relative
differences to the baseline policy, \texttt{Hash}. Here we measure the total runtime of an application
using \algname as the partitioning strategy in three modes, \emph{vertex partitioning} (one-dimensional balance
on vertex count), \emph{edge partitioning} (balance on edge count), and \emph{vertex-edge partitioning} (two-dimensional balance
both on vertex and edge counts). Every algorithm is applied in two configurations, \emph{small} and \emph{large}.
The first one uses the \texttt{FB-80B} graph and a cluster with $16$ worker machines, while the second one
process \texttt{FB-400B} using $128$ workers.

The key finding is that one-dimensional partitioning cannot provide consistent benefits across all the Giraph applications.
In fact, we observe performance regression for some instances, in particular, when the number of utilized worker machines
is large, that is, $k=128$. In this scenario, we notice a few workers whose running time is significantly larger than the average; see 
Figure~\ref{fig:hist}. Since in Giraph (and other vertex-centric systems) the computation is split into a number of supersteps that end with a 
global synchronization barrier, the performance is determined by the slowest worker.
Notice that a similar phenomena regarding the \emph{vertex partitioning} has been observed in earlier works~\cite{GLGBG12,AKCV18,Sun18,GHCIE17}.
In contrast, the two-dimensional partitioning always results in a speedup over the default \texttt{Hash}
strategy. The improvement is in the order of $10\%-30\%$ for the examined applications.

\begin{table}[tb]
  \small
  \centering
  \caption{Impact of partitioning policy on the running time and the amount of sent messages
    across $128$ Giraph workers for the \texttt{Page Rank} application applied on the \texttt{FB-400B}
    graph. The numbers are average values over $30$ iterations.}
  \begin{tabular}{lcrcccc}
    \toprule
    \multicolumn{1}{c}{Partitioning}
    & \multicolumn{3}{c}{Runtime, sec}
    & \multicolumn{3}{c}{Communication, GB} \\
    & \multicolumn{1}{c}{mean} & \multicolumn{1}{c}{max} & \multicolumn{1}{c}{stdev} & 
    \multicolumn{1}{c}{mean} & \multicolumn{1}{c}{max} & \multicolumn{1}{c}{stdev} \\
    \midrule
    \texttt{Hash}
    & $95$ & $102$ & $27$ & $69.5$ & $69.6$ & $2.4$ \\
    \texttt{vertex}
    & $93$ & $143$ & $25$ & $18.6$ & $47.6$ & $6.8$ \\
    \texttt{edge}
    & $82$ & $120$ & $22$ & $25.7$ & $38.2$ & $5.9$ \\
    \texttt{vertex-edge}
    & $84$ & $88$ & $21$ & $29.1$ & $30.6$ & $2.8$ \\
    \bottomrule
  \end{tabular}
  \label{table:pr_128}
\end{table}

To get a deeper understanding of the source of performance differences, we analyze the detailed logs
for the \texttt{Page Rank} application using a cluster with $128$ worker machines. Table~\ref{table:pr_128} shows
the measurements of the mean, maximum, and standard deviation of the time to compute a superstep by all the workers.
The results indicate that the with hash partitioning the workers are idling on average for $7$ seconds per
superstep waiting for the slowest one to complete the work. With one-dimensional partitioning the idling time is much longer,
$50$ seconds for vertex-based partitioning and $38$ seconds for edge-based one, which is the primary
reason for the performance regression. The two-dimensional partitioning results in a more even load across the
workers delivering a $13.2\%$ speedup. 
Table~\ref{table:pr_128} also indicates a significant communication reduction over the baseline partitioning, as
measured by the total size of messages sent between the workers via network. For the \texttt{Page Rank} application,
the average reduction is correlated with the edge locality of the corresponding partitioning.
However, an unbalanced partitioning causes some workers to use more memory resources and become a bottleneck for graph processing.

Finally, we emphasize that the timings analyzed in the section exclude the running times of the partitioner itself.
This is realistic for our use case in which the same friendship graph is expected to be processed multiple times for various analytics tasks.
Thus, the extra overhead incurred by a partitioning strategy is amortized among several runs. 

\subsection{Parameters of GD}
\label{sec:params}
In this section we perform an experimental comparison of various choices of the projection step algorithm in \algname and study its convergence properties.
Unless specified otherwise, we use two-dimensional \algname in the following setting:
1) balance is required with respect to the number of vertices and their degrees,
2) in the projection step we use ``one-shot'' alternating projection (see Section~\ref{sec:impl_projection}),
3) we use adaptive step size and vertex fixing as described in Section~\ref{sec:adapt}.


\begin{figure}[!tb]
  \centering
  \begin{subfigure}[b]{0.23\textwidth}
    \centering
    \includegraphics[width=\textwidth]{experiments/steps/LiveJournal_steps.pdf}
  \end{subfigure}
  \begin{subfigure}[b]{0.23\textwidth}
    \centering
    \includegraphics[width=\textwidth]{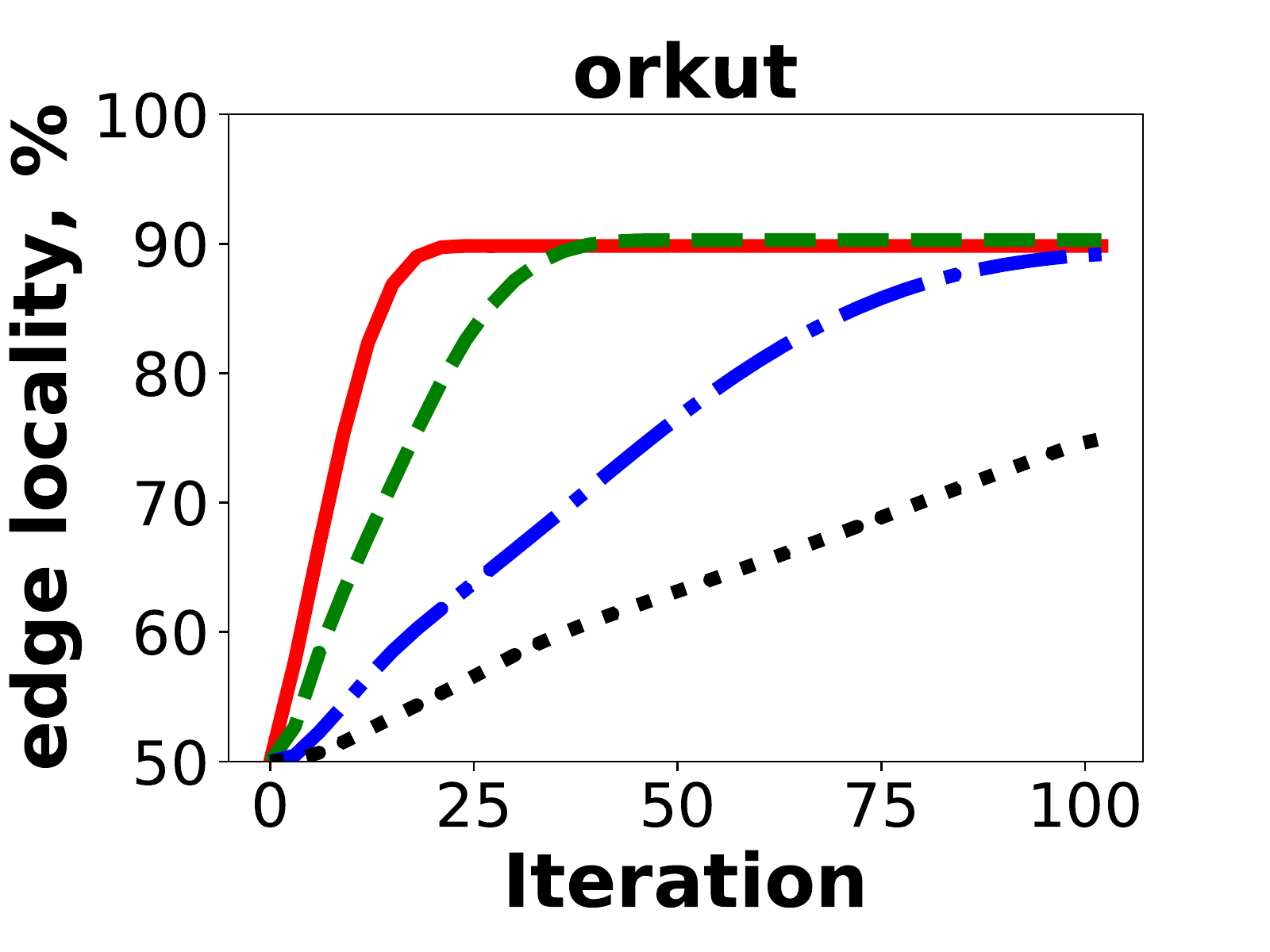}
  \end{subfigure}
  \begin{subfigure}[b]{0.46\textwidth}
    \centering
    \includegraphics[width=\textwidth]{experiments/steps/legend_steps.pdf}
  \end{subfigure}
  \caption{Comparison of step choices for \algname with fixed step length, that is, $\|\xt{t} - \xt{t+1}\|_2= const$, for $100$ iterations and $\xi = \sqrt{n}/100$. Step length $2 \cdot \xi$ results in good performance. }
  \label{fig:step}
\end{figure}

Since behavior of gradient descent algorithms can depend on selection of the step size parameters, we used experiments to establish convergence of \algname with different choices of these parameters.
In particular, our implementation aims to ensure that the step length $\|\xt{t} - \xt{t+1}\|_2$ remains close to constant between iterations.
A natural scaling parameter for the step length is $\sqrt{n}$ as it corresponds to the distance between the initial solution $\vx_0 = 0$ and any integral solution of the form $\{-1,1\}^n$.
As we show in Figure~\ref{fig:step} for various graphs a good choice of step size turns out to be $2 \frac {\sqrt n}{100}$, where $100$ is the limit we set on the number of iterations due to the constraints on the runtime during the execution.

\begin{figure}[!tb]
  \centering
  \begin{subfigure}[b]{0.23\textwidth}
    \centering
    \includegraphics[width=\textwidth]{experiments/adapt/LiveJournal_adapt.pdf}
  \end{subfigure}
  \begin{subfigure}[b]{0.23\textwidth}
    \centering
    \includegraphics[width=\textwidth]{experiments/adapt/LiveJournal_imb.pdf}
  \end{subfigure}
  \begin{subfigure}[b]{0.23\textwidth}
    \centering
    \includegraphics[width=\textwidth]{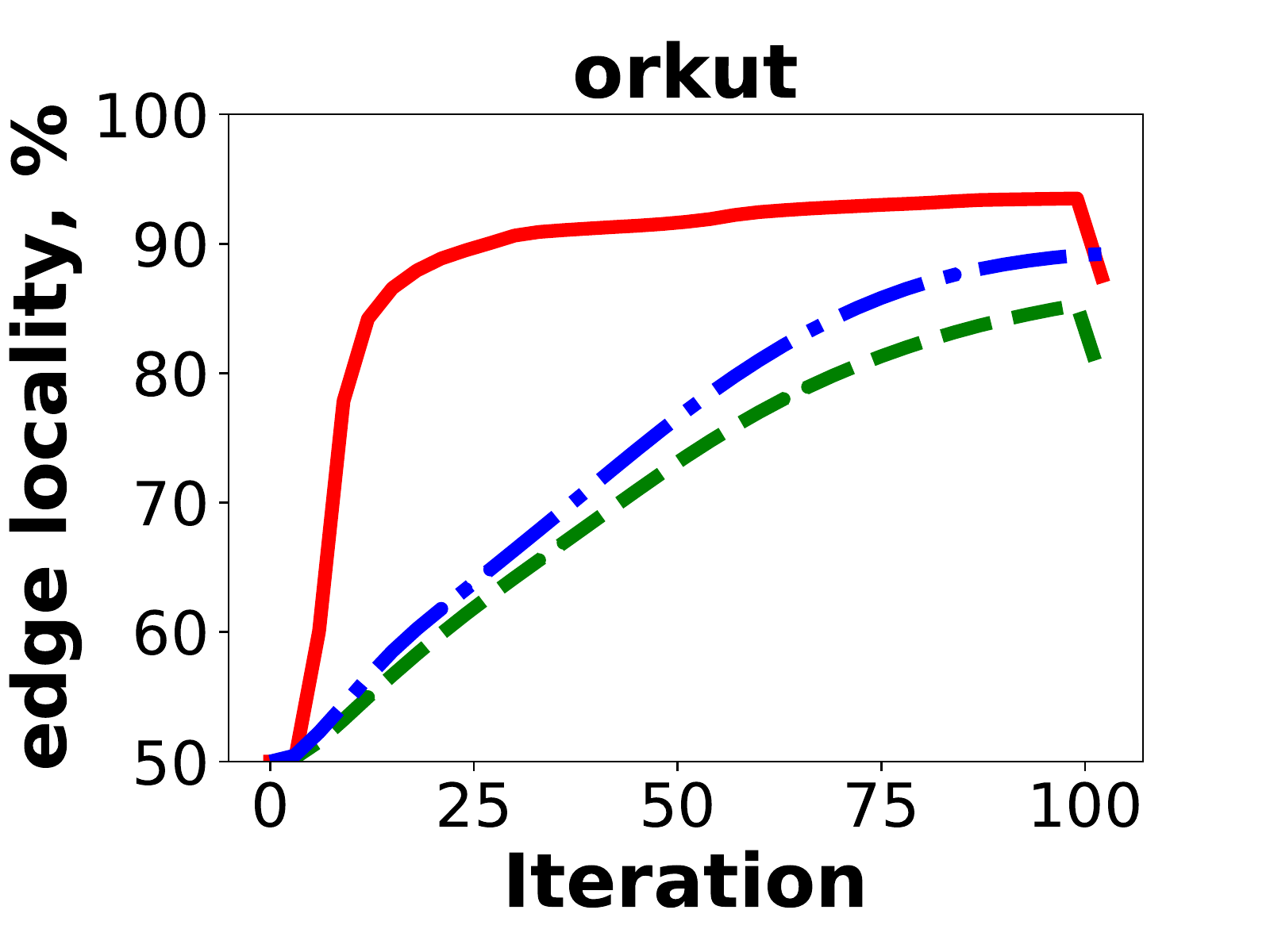}
  \end{subfigure}
  \begin{subfigure}[b]{0.23\textwidth}
    \centering
    \includegraphics[width=\textwidth]{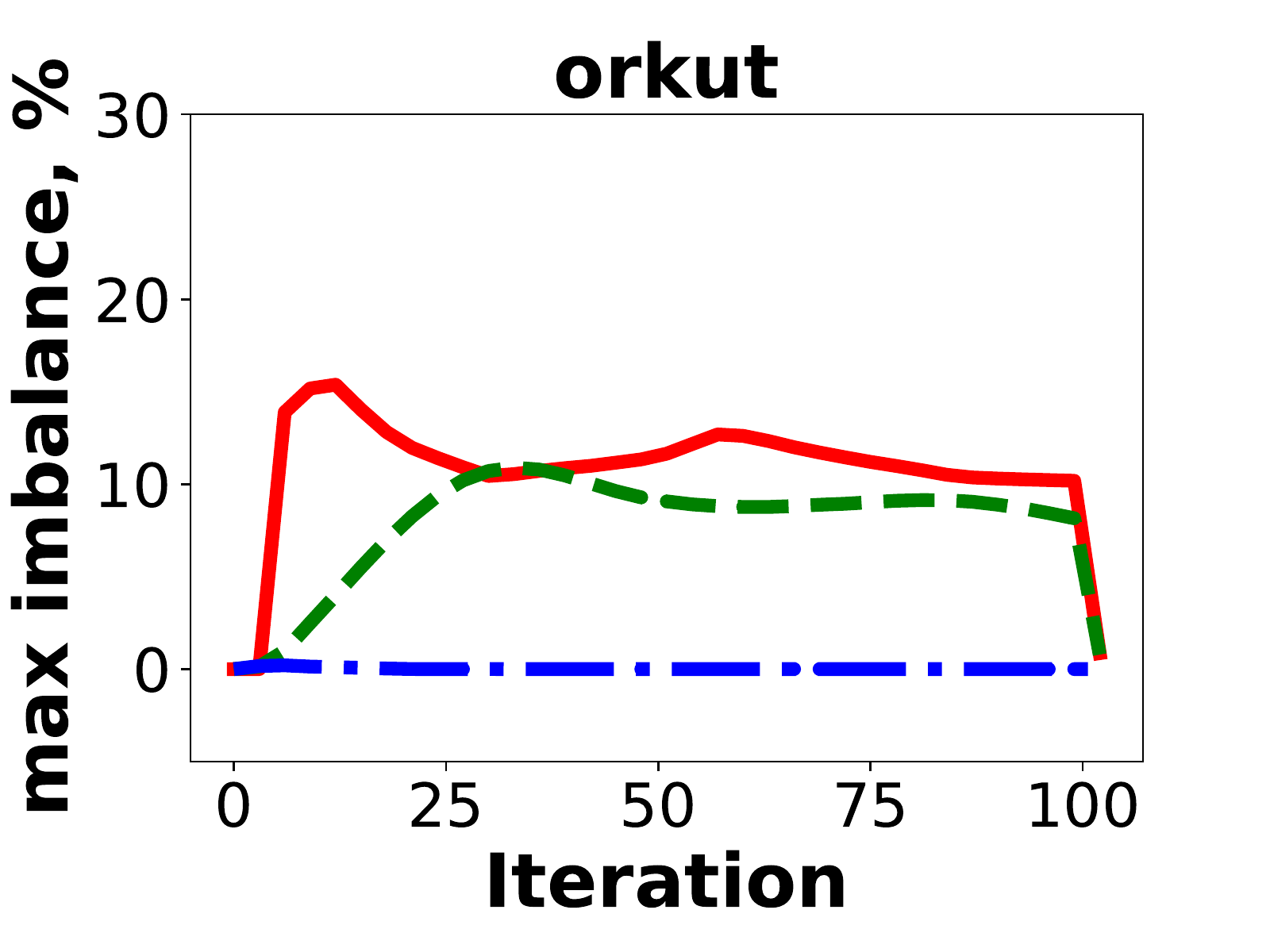}
  \end{subfigure}
  \begin{subfigure}[b]{0.46\textwidth}
    \vspace{2mm}
    \centering
    \includegraphics[width=\textwidth]{experiments/adapt/legend_adapt.pdf}
  \end{subfigure}
  \caption{Quality and imbalance comparison of \algname 1) without adaptive step size, 2) with adaptive step size and 3) with adaptive step size with vertex fixing.
  The left side shows edge locality and the right side~-- maximum imbalance over all dimensions.
    For nonadaptive and adaptive strategies the changes in the number of cut edges and imbalance in the last iteration are due to fixing in the end of the algorithm the accumulated imbalance resulting from ``one-shot'' alternating projection.
    Using \algname with adaptive step size and vertex fixing results in better locality and preserves almost perfect balance during algorithm execution.
  }
  \label{fig:adapt}
\end{figure}

In Figure~\ref{fig:adapt} we show how adaptive step size and vertex fixing described in  Section~\ref{sec:implementation} affect the performance of the algorithm. Note that compared with other methods vertex fixing not only improves quality but also preserves almost perfect balance even when simple ``one-shot'' alternating projection is used.
Finally, in Figure~\ref{fig:proj} we show analysis of performance of the algorithm under different choices of the projection method. 
\rlabel{line:R2.10}{The results show that the exact projection algorithm with sufficiently large allowed imbalance leads to the best performance. Larger imbalance permits more partitions, possibly including ones with better locality, allowing the overall algorithm to find partitions with better locality. However, the alternating projections algorithm can often be used to achieve similar performance. This is most likely due to the fact that the alternating projections algorithm despite not computing the projection outputs a point close enough to it.

}
\begin{figure}[!tb]
  \centering
  \begin{subfigure}[b]{0.23\textwidth}
    \centering
    \includegraphics[width=\textwidth]{experiments/proj/LiveJournal_proj.pdf}
  \end{subfigure}
  \begin{subfigure}[b]{0.23\textwidth}
    \centering
    \includegraphics[width=\textwidth]{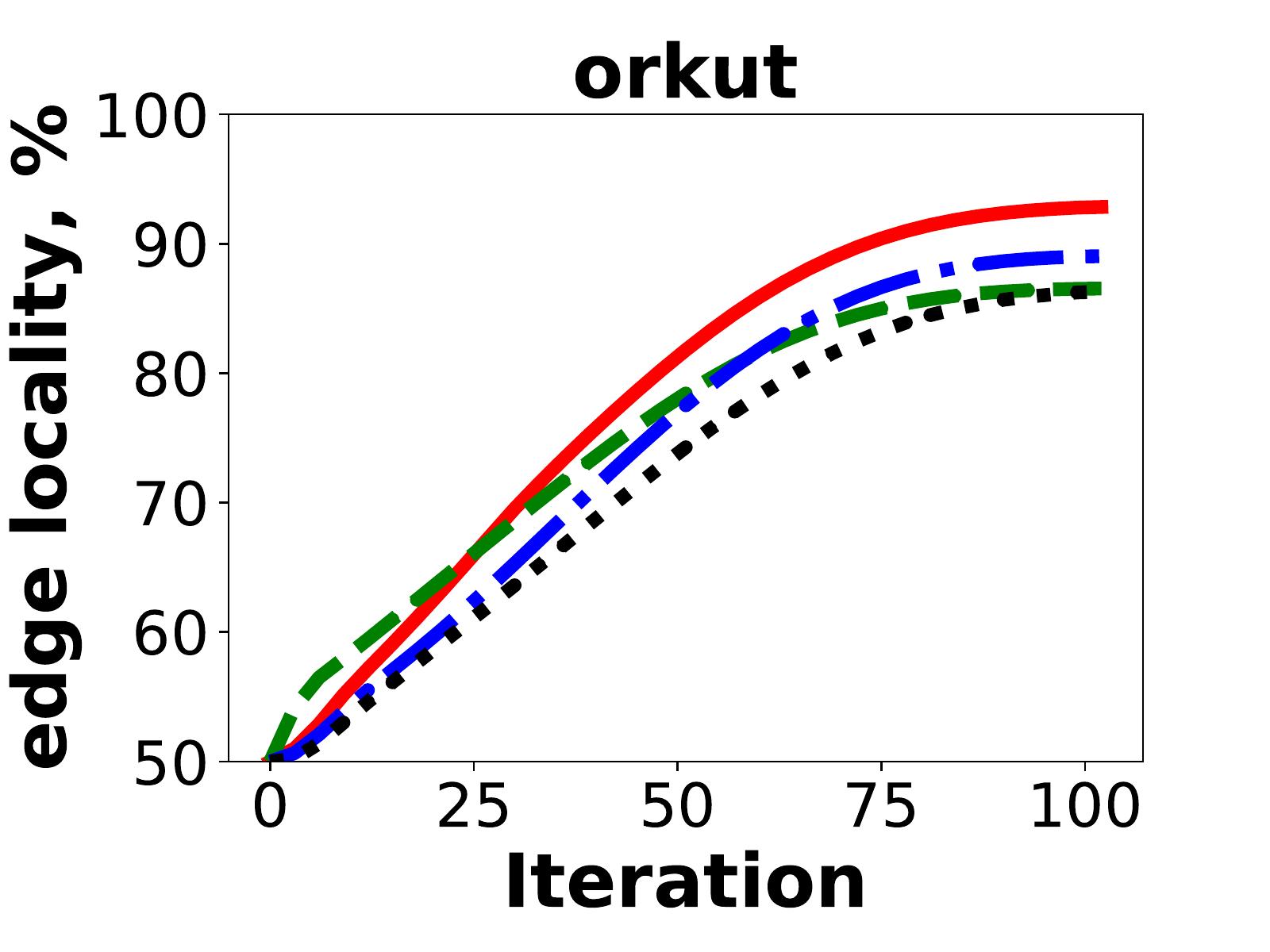}
  \end{subfigure}
  \begin{subfigure}[b]{0.46\textwidth}
    \centering
    \includegraphics[width=\textwidth]{experiments/proj/legend_proj.pdf}
  \end{subfigure}
    \caption{Quality comparison of \algname with various projection methods.
      We compare exact projection with various allowed imbalance parameters and ``one-shot'' alternating projection. Allowing more imbalance typically results in partitions with better quality.
      ``One-shot'' alternating projection, which we choose as our default implementation option due to its efficiency on larger datasets, produces partitions comparable with exact projection.
      Dykstra's projection produces the same results as the exact projection, and therefore is not shown.
    }
    \label{fig:proj}
\end{figure}


\subsection{Performance Analysis}

\begin{table}[ht]
  \begin{minipage}[b]{\linewidth}
    \centering
    \includegraphics[width=0.7\textwidth]{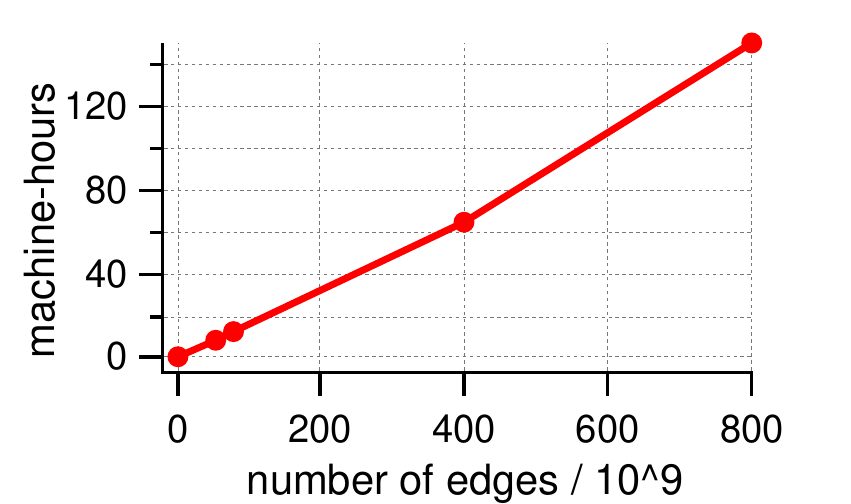}
    \vspace{0.05cm}
    \captionof{figure}{Scalability of the distributed implementation of \algname on 
      \texttt{FB-X} graphs of various size.
      The results indicate linear dependence of time on the number of edges.}
    \label{fig:mh}
  \end{minipage}
\end{table}

Finally, we analyze scalability of our algorithm. Our results are obtained on 
a Hadoop cluster of $128$ workers; each of the machines is a dual-node 2.4 GHz
Intel Xeon E5-2680 with 256GB RAM.
Figure~\ref{fig:mh} reports the running time of \algname in machine-hours
on \texttt{FB-X} graphs of various size with balance on two dimensions.
We observe a near-linear growth of the running time with the size of the input graph.
In comparison, the running time of the \texttt{SHP} algorithm exceed the values by a factor
of $1.5-2$ on the same cluster configuration. Despite the fact that our implementation is not 
specifically optimized for performance, \algname processes huge graphs within a few hours
in the distributed setting.

\section{Conclusion}\label{sec:conclusions}

We introduced a new \textsc{Multi-Dimensional Balanced Graph Partitioning} algorithm which produces balanced partitions according to multiple user-specified weight functions while maintaining high edge locality.
Our results show that this algorithm is scalable and for large graphs with small allowed vertex and edge imbalance outperforms existing solutions. Resulting partitions allow one to achieve substantial speedups in computational time for various computational tasks. This is in contrast with balancing on just one dimension (for example, vertex or edge count, separately), which can sometimes result in worse performance.
We state several open problems below.

One of the most interesting directions for future work is incorporating a wider range of balancing requirements, for example, those that can depend on the resulting partitioning itself such as the number of local edges and the maximum number of edges going between any pair of parts in the resulting partition. For example, the latter quantity can substantially affect performance of distributed computation tasks in Giraph-like systems as communication between different machines depends on the number of edges between them. Note that our proposed algorithm can't directly handle such solution-dependent weight functions as they can't be specified through an a priori fixed collection of weight functions.

A scalable algorithm for solving multi-dimensional balanced partitioning into $k$ parts without using recursive partitioning. As discussed in Section~\ref{sec:k-partitioning}, applying similar algorithm to straightforward problem relaxation results into $\Oh(k \cdot |E|)$ communication, which comes from inherently continuous nature of the algorithm compared to discrete ones.
In discrete algorithms a vertex can occupy only one bucket, but in our algorithm it can occupy all buckets with some probabilities. Since all these probabilities may change, $\Theta(k)$ information can be sent to neighbors.

An interesting theoretical question is finding a fast algorithm for exact projection for $d > 2$.
As we will show in
\iffull Appendix~\ref{app:bin_search}, \else in the full version, \fi it is possible to use nested binary search to find $\set{\lambda_j}$ (and therefore the projection) with arbitrary precision.
Unfortunately, the running time of the suggested algorithm is unknown, because it is unclear how to estimate left and right bounds for binary search.
Determining these bounds gives an algorithm with running time $\Oh(n \cdot \prod_{i=1}^d \log \frac{r_j - l_j}{\delta})$, where $l_j$ and $r_j$ are bounds for $\lambda_j$ and $\delta$ is the required precision.

Another interesting theoretical question is understanding the convergence properties of our algorithm (or a similar gradient descent based method) under some assumption about the spectral properties of the graph. We see this as a challenging open problem -- while noisy gradient descent is known to have fast convergence for non-convex optimization subject to equality constraints, if inequality constraints are allowed convergence analysis is unknown~\cite{GHJY15}.

\appendix
\iffull
\section{Multidimensional projection}

In this section we consider projection problem in multidimensional case.
In section~\ref{SEC:PROJECTION} we reduced projection to the following optimization problem.
\eqbox{\begin{align*}
& \text{Minimize:}   && f(\vx) = \|\vy - \vx\|_2^2 && \\ 
& \text{Subject to:} && g_i = x_i^2 - 1 \le 0      && \forall i \in [n]\\
&                    && \sum_{i = 1}^n \wj_i x_i = \epsilon && \forall j \in S_+; \\
&                    && \sum_{i = 1}^n \wj_i x_i = -\epsilon && \forall j \in S_-
\end{align*}}

Then KKT conditions for this problem were further reduced to the following problem. Given $\vy \in \mathbb R^n$ we need to find its projection $\vx$ whose coordinates are given as $x_i(\lambda_1, \ldots, \lambda_d) = [y_i - \sum_j \wj_i \lambda_j]$ by selecting the values $(\lambda_1, \ldots, \lambda_d)$ in order to satisfy the balance constraints, i.e. $\sum_{i = 1}^n \wj_i x_i = \epsilon$ for $j \in S_+$ and $\sum_{i = 1}^n \wj_i x_i = -\epsilon$ for $j \in S_-$.
We consider more general constraints: $\sum_{i = 1}^n \wj_i x_i = c_j$ for $j \in [d]$, where $\set{c_j}$ are some constants.
Let $\vlam = (\lambda_1, \ldots, \lambda_d)$.
Since $\vx$ can be computed based on $\vlam$, it remains to show how to find $\vlam$ satisfying these constraints.

The contents of this section are the following:
\begin{itemize}
	\item We show that it's possible to find $\vlam$ (and therefore $\vx$) with arbitrary precision using nested binary search.
	\dtodo{But we don't know bounds}\gtodo{Is this even in the paper?}
	\item We describe an $\Oh(n \log n)$-time algorithm finding the exact values of $\vlam$ in $2$-dimensional case.
\end{itemize}

\subsection{Nested binary search}
\label{app:bin_search}

Recall that $\hj(\vlam) = \sum_i w_{ij} x_i$ .
As shown in Section~\ref{SEC:PROJECTION}, $\hj(\vlam) = \sum\limits_{i=1}^n \hj_i(\vlam)$, where:
\begin{align*}
\hj_i(\vlam) = \begin{cases}
	\wj_i  & \text{if } \sum\limits_k \lambda_k \wk_i < y_i - 1 \\ 
	\wj_i (y_i - \sum\limits_k \lambda_k \wk_i) & \text{if } \sum\limits_k\lambda_k \wk_i \in [y_i - 1, y_i + 1]\\
	-\wj_i & \text{if } \sum\limits_k \lambda_k \wk_i > y_i + 1 \\
\end{cases}
\end{align*}
We want to find $\vlam^*$ such that $\hj(\vlam^*) = c_j$ for all $j \in [d]$.
\begin{lemma}[Uniqueness]
	There is at most one point $\vx$ for which there exists $\vlam^*$ such that $(\vx, \vlam^*)$ satisfy KKT conditions.
\end{lemma}
\begin{proof}
	Our optimization problem is convex, since $L_2$-norm is a convex function, cube and planes are convex sets and their intersection is also a convex set. As follows from~\cite{BV04}, any pair $(\vx, \vlam^*)$ satisfying KKT conditions is a solution (i.e. $\vx$ is the projection). By strict convexity of $L_2$-norm the projection is unique, and therefore there is at most one $\vx$ satisfying KKT conditions.
\end{proof}
Note that there can be several $\vlam^*$ corresponding to the same $\vx$.
In the rest of the section we show that it is possible to find $\vlam^*$ using nested binary search.
For that purpose we define auxiliary functions $\Delta_1, \ldots, \Delta_d$ in the following way.

For any value of $\lambda_1$ we would like to find $\lambda_2, \ldots, \lambda_d$ such that constraints $\h 2(\vlam)=c_2, \ldots, \h d(\vlam)=c_d$ are satisfied.
We define $\Delta_1(\lambda_1)$ as $\h 1(\vlam)$.
We will show that $\Delta_1$ is well-defined (when the feasible space is not empty) and monotone.
Therefore, we can use binary search to find $\lambda_{1}$ for which $\h 1(\vlam)=c_1$ is satisfied.

Consider the nested problem. Assume that $\lambda_1$ is fixed.
Then for any value of $\lambda_2$ we would like to find $\lambda_3, \ldots, \lambda_d$ such that constraints $\h 3(\vlam)=c_3, \ldots, \h d(\vlam)=c_d$ are satisfied.
Similar to $\Delta_1$ we define $\Delta_2(\lambda_1, \lambda_2)$ as $\h 2(\vlam)$ and we will show that $\Delta_2$ is well-defined and monotone on $\lambda_2$.
Therefore, again, we can use binary search to find $\lambda_2$.
We define $\Delta_t(\lambda_1, \ldots, \lambda_t)$ for all $t$ and show that $\Delta_t$ is monotone on $\lambda_t$.

\begin{definition}
	Consider $t \in [d]$. Let $\vlam = (\lambda_1, \dots, \lambda_d)$ and assume that constraints $\hj(\vlam) = c_j$ are satisfied for all $j > t$. Then we define $\Delta_t(\lambda_1, \ldots, \lambda_t) \triangleq \h t(\vlam)$ and call $\lambda_{t+1}, \ldots, \lambda_d$ \emph{suitable} for $\lambda_1, \ldots, \lambda_t$.
\end{definition}
Note that $\Delta_t$ is a function of the first $t$ coordinates.
\begin{lemma}[$\Delta$ is well-defined]\label{app:lem:well_def}
	For fixed $\lambda_1, \ldots, \lambda_t$ different suitable $\lambda_{t+1}, \ldots, \lambda_d$ produce the same $\vx(\vlam)$.
	Therefore, $\Delta_t(\lambda_1, \ldots, \lambda_t)$ is the same for different suitable $\lambda_{t+1}, \ldots, \lambda_d$.
	If the feasible space is not empty, then for fixed $\lambda_1, \ldots, \lambda_t$ there exist suitable $\lambda_{t+1}, \ldots, \lambda_d$.
\end{lemma}
\begin{proof}
		Fix $\lambda_1, \ldots, \lambda_t$. Denote $y_i' = y_i - \sum\limits_{j \le t} \lambda_j w_j$. Then we obtain the following problem: find $\lambda_{t+1}, \ldots, \lambda_d$, such that $\vx = [\vy' - \sum_{j > t} \lambda_j \wj]$ and $\sum_{i = 1}^n \wj_i x_i = c_j$ for all $j > t$. Therefore, we reduced the problem to $(d-t)$-dimensional problem of the same form, and by Uniqueness Lemma there exists exactly one $\vx$, satisfying all constraints.
\end{proof}

\begin{lemma}[Solution convexity]
	The set of $\vlam$ such that $(\vx, \vlam)$ is KKT solution is convex.
\end{lemma}
\begin{proof}
	By Uniqueness Lemma there is at most one $\vx$ satisfying KKT. Consider two KKT solutions $(\vx, \vlam)$ and $(\vx, \vlam')$.
    Therefore \[\vx = [\vy - \sum_j \wj \lambda_j] =  [y - \sum_j \wj \lambda_j']\]
    We will show that $(\vx,\ \alpha \vlam + (1 - \alpha) \vlam')$ is also a solution for any $\alpha \in [0; 1]$.
	For each $x_i$ consider $3$ cases depending on rounding of $x_i$:
	\begin{enumerate}
		\item $x_i = 1$. Then $\sum_j \wj_i \lambda_j \le y_i - 1$ and $\sum_j \wj_i \lambda_j' \le y_i - 1$.
		By multiplying the first inequality by $\alpha$ and the second one by $(1 - \alpha)$ and then summing them up we obtain
		\[\sum_j \wj_i (\alpha \lambda_j + (1 - \alpha)\lambda_j') \le y_i - 1\]
		\item $x_i = -1$. Similar to the first case.
		\item $x_i \in (-1; 1)$. $\sum_j \wj_i \lambda_j = y_i - x_i$ and $\sum_j \wj_i \lambda_j' = y_i - x_i$. Therefore, 
		\[\sum_j \wj_i (\alpha \lambda_j + (1 - \alpha)\lambda_j') = y_i - x_i\]
	\end{enumerate}
	\mbox{} 
\end{proof}

\begin{lemma}
	$\Delta_t$ is continuous
\end{lemma}
\begin{proof}
    Follows from the fact that projection is continuous function of the original point.
    For small enough $\eps_j$ the projection of $\vy - \sum_{j \le t} \lambda_j \wj$ is close to projection of $\vy - \sum_{j \le t} (\lambda_j + \eps_j) \wj$, and so are their values of $\hj$, $j > t$.
\end{proof}

\begin{theorem}[$\Delta_t$ monotonicity]\label{thm:monotonicity}
	Consider two points $(\lambda_1,\ldots,\lambda_{t-1},\lambda_t')$ and $(\lambda_1,\ldots,\lambda_{t-1},\lambda_t'')$ such that
\[\Delta_t(\lambda_1,\ldots,\lambda_{t-1},\lambda_t') = \Delta_t(\lambda_1,\ldots,\lambda_{t-1},\lambda_t'') = C.\]
Then for any $\alpha \in [0; 1]$
\[\Delta_t(\lambda_1,\ldots,\lambda_{t-1},\ \ \alpha \lambda_t' + (1- \alpha) \lambda_t'') = C.\]
Since $\Delta_t$ is continuous, $\Delta_t$ is monotone on $\lambda_t$.
\end{theorem}
\begin{proof}
	Since
	\[\Delta_t(\lambda_1,\ldots,\lambda_{t-1},\lambda_t') = \Delta_t(\lambda_1,\ldots,\lambda_{t-1},\lambda_t'') = C,\]
	there exist $\lambda_{t+1}', \ldots, \lambda_n'$ and  $\lambda_{t+1}'', \ldots, \lambda_n''$ such that
	\begin{align*}
	\h t(\vlam') &= \h t(\vlam'') = C \\
	\hj(\vlam')  &= \hj(\vlam'') = c_j \text{ for all $j > t$},
	\end{align*}
    where $\vlam' = (\lambda_1,\ldots,\lambda_{t-1},\lambda_t',\ldots,\lambda_d')$ and $\vlam'' = (\lambda_1,\ldots,\lambda_{t-1},\lambda_t'',\ldots,\lambda_d'')$.
    
    Denote $y_i' = y_i - \sum\limits_{j < t} \wj_i \lambda_j$. Consider the following problem: find $\lambda_{t+1}, \ldots, \lambda_d$, such that
	\begin{align*}
	& \vx = \vy' - \sum_{j \ge t} \lambda_j \wj \\
    & \sum_{i = 1}^n \w{\mathbf{t}}_i x_i = C \\
    & \sum_{i = 1}^n \wj_i x_i = c_j \text{ for all } j > t
	\end{align*}
    We obtained $(d-t+1)$-dimensional problem.
    Both points are solutions to this problem, and by Convexity lemma the set of its solution is convex.
\end{proof}

As follows from Theorem~\ref{thm:monotonicity}, if the projection exists then it's possible to find $\vlam^*$ with arbitrary precision using nested binary search on each coordinate.
Unfortunately, it's unclear how to estimate binary search bounds.
While it's possible to find them by expanding the bounds until they contain the solution, the resulting running time becomes unknown.

\fi

\iffull
\subsection{Projection for D = 2}
\else
\section{Projection for D = 2}
\fi
\label{app:projection2d}
In this section we introduce a randomized $O(n \log n)$-time algorithm for finding projection for $d=2$.
Recall from Section~\ref{SEC:PROJECTION} that for $\vy \in \mathbb{R}^n$ we need to find $\vlam^* = (\lambda_1^*, \lambda_2^*)$ such that $\h 1(\vlam^*) = c_1$ and $\h 2(\vlam^*) = c_2$.
For $\vlam=(\lambda_1, \lambda_2)$ we define $\hj(\vlam) = \sum\limits_{i=1}^n \hj_i(\vlam)$ for $j \in \{1,2\}$, where
\begin{align*}
\hj_i(\vlam) = \begin{cases}
\wj_i  & \text{if } \sum\limits_k \lambda_k \wk_i < y_i - 1 \\ 
-\wj_i & \text{if } \sum\limits_k \lambda_k \wk_i > y_i + 1 \\
\wj_i (y_i - \sum\limits_k \lambda_k \wk_i) & \text{otherwise.}\\
\end{cases}
\end{align*}
Once we find $(\lambda_1^*, \lambda_2^*)$ we can compute the coordinates of $\vx$ as $x_i = [y_i - \w 1_i \lambda_1^* - \w 2_i \lambda_2^*]$.
We introduce an auxiliary function $\Delta$
\iffull(corresponding to $\Delta_1$ from the previous section)\fi
which we use to solve the above problem using binary search:
\ifnotfull\footnote{In the full version of the paper we will show that the $\Delta$ is well-defined and monotone.}\fi
\begin{definition}
	Suppose that $\lambda_1$ is such that there exists $\lambda_2$ for which the constraint $\h 2(\lambda_1, \lambda_2) = c_2$ is satisfied.
    Then we define $\Delta(\lambda_1) \triangleq \h 1(\lambda_1, \lambda_2)$.
\end{definition}

We now describe an $\Oh(n \log n)$-time algorithm for finding $(\lambda_1^*, \lambda_2^*)$.
The algorithm is shown as Algorithm~\ref{alg:2D}.
It takes as a parameter a Boolean variable $\Delta^+$ indicating whether $\Delta$ is an increasing or decreasing function. We run the algorithm under both assumptions and select a solution satisfying the constraints.
\begin{algorithm}[h]
\SetKwInOut{Input}{input}
	\SetKwInOut{Output}{output}
	\SetKwRepeat{Do}{do}{while}
    \SetNoFillComment
	\Input{$\set{\w 1_i}$, $\set{\w 2_i > 0}$, $\set{y_i}$, $c_1$, $c_2$}
	\nonl \textbf{parameter:} $\Delta^+ \in \{true, false\}$ indicating whether $\Delta$ increases\\
	\Output{$(\lambda_1^*, \lambda_2^*)$}
    \SetKwProg{Fn}{Function}{}{}
	\caption{Function returning $\lambda_1^*, \lambda_2^*$ for given $2$-dimensional problem.
    }
    \label{alg:2D}
    \Fn{Project-2D}{
        \tcc{$L = $ set of lines parameterized by $(y_i, \w 1_i, \w 2_i, \pm 1)$ corresponding to lines of the form $y_i - \lambda_1 \w 1_i - \lambda_2 \w 2_i = \pm 1$} 
        $L := \set{ (y_i, \w 1_i, \w 2_i, \pm 1) | i \in [n]}$ \label{lst:line:find_slope_begin} \label{lst:line:l-def} \\
        $\lambda_1^l := -\infty$, $\lambda_1^r := +\infty$ \\
        \vspace{3mm}
        \tcc{Run binary search}
        \While {true} {
        	\texttt{$\Lambda :=$ set of intersection points $(\lambda_1, \lambda_r)$ of lines in $L$ such that $\lambda_1 \in (\lambda_1^l, \lambda_1^r)$} \\
	    	\If{$\Lambda = \emptyset$} { \texttt{break}}
            
            \texttt{Sample a uniformly random intersection point $(\lambda_1',\lambda_2')$ from $\Lambda$ \label{lst:line:sample}}\\
            \eIf{$\Delta(\lambda_1') > c_1$ \label{lst:line:if-begin}} {
            	If $\Delta^+$ set $\lambda_1^r := \lambda_1'$,  otherwise set $\lambda_1^l := \lambda_1'$
            } {
            	If $\Delta^+$ set $\lambda_1^l := \lambda_1'$,  otherwise set $\lambda_1^r := \lambda_1'$
            	\label{lst:line:if-end} \label{lst:line:find_slope_end}
            }
        }
        \vspace{3mm}
        \texttt{Let $\set{R_t}_{t=1}^T$ be a partition of $(\lambda_1^l, \lambda_1^r) \times \mathbb{R}$ by boundary lines (see Fig.~\ref{fig:partition}), sorted from bottom to top}
        \label{lst:line:find_region_begin}\\
        \texttt{Compute coefficients for the system of linear equations for $R_1$ (as in Theorem~\ref{thm:find_region})}  \\
        \For {$t = 1 \dots T$} {
        	\texttt{Let $(\lambda_1, \lambda_2)$ be a solution to the linear system for $R_t$} \\
        	\If {$(\lambda_1, \lambda_2) \in R_t$} {
            	\Return{$(\lambda_1, \lambda_2)$}
            }
            \texttt{Update the coefficients corresponding to crossing the boundary line between $R_t$ and $R_{t + 1}$ as shown in Theorem~\ref{thm:find_region}} \label{lst:line:find_region_end}
        }
    }
\end{algorithm}

We outline the main ideas behind Algorithm~\ref{alg:2D} below.
Consider the $(\lambda_1, \lambda_2)$ plane partitioned by the following lines (which we call \emph{boundary lines}):
\begin{align*}
& y_i - \lambda_1 \w 1_i - \lambda_2 \w 2_i  = 1 \\
& y_i - \lambda_1 \w 1_i - \lambda_2 \w 2_i = -1,
\end{align*}
for all $i$.
Let $L$ be the set of boundary lines (line~\ref{lst:line:l-def}). We refer to the subsets of the plane resulting from its partition by the boundary lines as \emph{regions} (see Figure~\ref{fig:regions} where the regions are referred to as $\set{\mathtt{T_i}}$).
Boundary lines separate the plane into half-planes corresponding to the different cases in the definitions of the corresponding $\hj_i$.
Therefore, inside each region all $\hj_i$ are linear and hence $\hj$ are also linear.

The intuition behind the algorithm is then as follows (in order to achieve the best performance the exact details differ slightly from this simplified presentation).
Suppose we could find a region that contains some solution $\vlam^*$.
Then since constraint functions are linear inside the region, in order to find $\vlam^*$ we could solve a system of linear equations over $\lambda_1$ and $\lambda_2$.
We identify such region, with binary search over $\lambda_1$ by using monotnicity of $\Delta$.
We consider only a finite set of values: $\lambda_1$-coordinates of intersections of boundary lines.
Since there are $\Oh(n)$ boundaries, there are $\Oh(n^2)$ intersections(e.g., in Figure~\ref{fig:regions} we consider only points $a$, $b$, $c$ and $d$). Hence $\Oh(\log n)$ iterations of binary search suffice.
The only difference between Algorithm~\ref{alg:2D} and the above approach is that after the binary search on $\lambda_1$ we still have to try $\Oh(n)$ regions to identify the exact region which contains $\vlam^*$ (see Algorithm~\ref{alg:2D} for the details).

Now consider one iteration of the binary search.
Let $\lambda_1^l$ and $\lambda_1^r$ be its current boundaries. 
Let $\Lambda$ be a set of all intersection points $(\lambda_1, \lambda_2)$ such that $\lambda_1 \in (\lambda_1^l, \lambda_1^r)$.
Since $\Delta$ is monotone, for any $\lambda_1'$ we can use binary search by checking whether $\vlam^*$ is greater or less than $\lambda_1'$ through a comparison of $\Delta(\lambda_1')$ and $c_1$ (lines~\ref{lst:line:if-begin}-\ref{lst:line:if-end}).
Computing $\Delta(\lambda_1')$ requires solving the one-dimensional problem over $\lambda_2$ discussed in Section~\ref{sec:exact-projection} and thus can be done in $\Oh(n)$ time.

In order to have binary search run in $\Oh(\log n)$ iterations it suffices to be able to find a value $\lambda_1' \in (\lambda_1^l, \lambda_1^r)$ which with constant probability splits $\Lambda$ into two subsets of points, those with $\lambda_1 > \lambda_1'$ and with $\lambda_1 < \lambda_1'$ respectively, of size at most $\frac 23 n$ each.
In particular, it suffices to sample a uniformly random point $(\lambda_1', \lambda_2')$ from $\Lambda$. The following lemma bounds the overall running time of these sampling steps.
\begin{lemma}
The overall time required for sampling random points from $\Lambda$ in line~\ref{lst:line:sample} of Algorithm~\ref{alg:2D} is $\Oh(n \log n)$.
\end{lemma}
\begin{proof}
Consider three cases:
\begin{compactenum}
\item $|\Lambda| > n \log n$. In this case we sample $\Oh(n)$ uniformly random pairs of lines from $L$ and find an intersection of each pair (assume no parallel lines which can be handled separately).
Since the number of lines is $\Oh(n)$ w.h.p. we sample at least one intersection which lies in $\Lambda$.
The last condition can be checked in $\Oh(n)$ time and if it doesn't hold then we conclude that w.h.p. $|\Lambda| \le n \log n$.
We then compute $S$, the set of all points in $\Lambda$ in $\Oh(n \log n)$ time as described below and proceed to the second case.
\iffull

To find $\Lambda$ we first find intersections of all lines from $L$ with lines $\lambda_{1} = \lambda_{1}^l$ and $\lambda_{1} = \lambda_{1}^r$.
We call $\lambda_2$-coordinates of the intersection points \emph{events}.
Each line $\ell \in L$ creates two event: $\ell_{open}$ corresponds to smaller $\lambda_{2}$ and $\ell_{close}$~-- to the larger one.

Consider two lines $a$ and $b$ such that $a_{open} \ge b_{open}$.
These lines intersect in one of two cases.
If they are opened on different sides (i.e. one on $\lambda_1^l$ and another one~-- on $\lambda_1^r$), then $b_{open}$ should be greater than $a_{close}$, as shown in Figure~\ref{fig:intersect_dif}.
If they are opened on the same side, then it should be $b_{close} \ge a_{close}$, i.e. $[a_{open}, a_{close}] \subseteq [b_{open}, b_{close}]$, as shown in Figure~\ref{fig:intersect_same}.
\begin{figure}[!tb]
	\centering
	\begin{subfigure}[b]{0.23\textwidth}
		\centering
\begin{tikzpicture}[scale=0.5]
	\draw[thick,draw=red,name path=left] (-1, -2.5) -- (-1, 2.5) node[above] {\small $\lambda_1^l$};
	\draw[thick,draw=red,name path=right] (1, -2.5) -- (1, 2.5) node[above] {\small $\lambda_1^r$};
	\draw[name path=below] (-2, -2.5) -- (3.5, 2.5);
	\draw[name path=above] (-2.5, 2.5) -- (2.5, -2.5);
	
	\draw[fill=black, name intersections={of=above and left, by=x}] (x) circle (0.7mm) node[left] {$a_{close}$};
	\draw[fill=black, name intersections={of=above and right, by=x}] (x) circle (0.7mm) node[right] {$a_{open}$};  
	\draw[fill=black, name intersections={of=below and left, by=x}] (x) circle (0.7mm) node[left] {$b_{open}$};
	\draw[fill=black, name intersections={of=below and right, by=x}] (x) circle (0.7mm) node[right] {$b_{close}$};
\end{tikzpicture}
		\caption{Intersection of lines opened on different sides}
		\label{fig:intersect_dif}
	\end{subfigure}
	\begin{subfigure}[b]{0.23\textwidth}
		\centering
\begin{tikzpicture}[scale=0.5]
	\draw[thick,draw=red,name path=left] (-1, -2.5) -- (-1, 2.5) node[above] {\small $\lambda_1^l$};
	\draw[thick,draw=red,name path=right] (1, -2.5) -- (1, 2.5) node[above] {\small $\lambda_1^r$};
	\draw[name path=below] (-2, -2.5) -- (2.5, 2.5);
	\draw[name path=above] (-3.5, -0.8) -- (3.5, 0.2);
	
	\draw[fill=black, name intersections={of=above and left, by=x}] (x) circle (0.7mm) node[above left] {$a_{open}$};
	\draw[fill=black, name intersections={of=above and right, by=x}] (x) circle (0.7mm) node[below right] {$a_{close}$};  
	\draw[fill=black, name intersections={of=below and left, by=x}] (x) circle (0.7mm) node[left] {$b_{open}$};
	\draw[fill=black, name intersections={of=below and right, by=x}] (x) circle (0.7mm) node[right] {$b_{close}$};
\end{tikzpicture}
		\caption{Intersection of lines opened on the same side}
		\label{fig:intersect_same}
	\end{subfigure}
\end{figure}

We process all events in increasing order and 
for each side we maintain the set of lines opened on this side.
We sort lines in these sets by their closing events.
When event $\ell_{open}$ arrives, we find intersections of $\ell$ with opened lines in the following way.
To handle the first case, we intersect $\ell$ with all lines opened on the other side.
To handle the second case, we intersect $\ell$ with all lines opened on the same side and closing after $\ell_{close}$.
\else
We show how to compute $\Lambda$ in the full version.
\fi
\item $n \le |\Lambda| \le n \log n$.
Note that in this case $\Lambda = \set{(\lambda_1, \lambda_2) \in S | \lambda_1 \in (\lambda_1^l; \lambda_1^r)}$, where $S$ is as defined above.
We sample $\Oh(n)$ random points from $S$ so that w.h.p. we get at least one point from $\Lambda$.
As before, if this doesn't happen, we conclude that w.h.p. $|\Lambda| < n$ and proceed to the last case.
\item $|\Lambda| < n$. In this case we maintain $\Lambda$ directly.
When we sample a random point $(\lambda_1', \lambda_2') \in \Lambda$, we remove from $\Lambda$ all points on one of the side from $\lambda_1'$ as directed by the binary search.
\end{compactenum}
In each of the cases above one iteration can be implemented in $\Oh(n)$ time and pre-/post-processing between the cases takes $\Oh(n \log n)$ time.
Since there are $\Oh(\log n)$ iterations, sampling takes $\Oh(n \log n)$ time overall.
\end{proof}
\begin{figure}[!h]
	\centering
\begin{tikzpicture}[scale=0.9]
	\draw[thin,->,color=gray] (-3.5,0) -- (3.5,0) node[right] {\small $\lambda_1$} coordinate(x axis);
	\draw[thin,->,color=gray] (0,-2.5) -- (0,2.5) node[above] {\small $\lambda_2$} coordinate(y axis);
	\draw (-1.5, 2.5) -- (3.5, -2.5);
	\draw (-3.5, 2.5) -- (1.5, -2.5);
	\draw (-3.5, 1) -- (3.5, 1);
	\draw (-3.5, -1) -- (3.5, -1);
	\node at (-1, -1.5) {$\mathtt{T_1}$};
	\node at (1.5, -1.5) {$\mathtt{T_2}$};
	\node at (3, -1.5) {$\mathtt{T_3}$};
	\node at (-2.5, 0.5) {$\mathtt{T_4}$};
	\node at (-0.5, 0.5) {$\mathtt{T_5}$};
	\node at (1.5, 0.5) {$\mathtt{T_6}$};
	\node at (-3, 1.5) {$\mathtt{T_7}$};
	\node at (-1, 1.5) {$\mathtt{T_8}$};
	\node at (1.5, 1.5) {$\mathtt{T_9}$};
	\draw[thick,draw=red] (-1.3, -2.5) -- (-1.3, 2.5) node[above] {\small $\lambda_1^l$};
	\draw[thick,draw=red] (2.5, -2.5) -- (2.5, 2.5) node[above] {\small $\lambda_1^r$};
	\draw[fill=black] (-2,1) circle (0.7mm) node[above right] {$a$};
	\draw[fill=red] (0,1) circle (0.7mm) node[above right] {$b$};
	\draw[fill=red] (0,-1) circle (0.7mm) node[above right] {$c$};
	\draw[fill=red] (2,-1) circle (0.7mm) node[above right] {$d$};
\end{tikzpicture}
	\caption{
		Example of regions for $n=2$, $\vy = (0, 0)$. 
		The boundary lines are $\lambda_1 + \lambda_2 = \pm 1$ and $\lambda_2 = \pm 1$.
		These lines partition $(\lambda_1, \lambda_2)$-plane into nine regions $T_1, T_2, \dots, T_9$.
		Intersection points are $a$, $b$, $c$ and $d$.
		Current intersection points considered by the algorithm (those between $\lambda_1^l$ and $\lambda_1^r$) are shown in red.
	}
	\label{fig:regions}
\end{figure}

Using the above algorithm we can find $\lambda_1^l$ and $\lambda_1^r$ such that there are no intersection points between them.
Since there are $\Oh(\log n)$ iterations and each of them requires $\Oh(n)$ time on average, the total running time is $\Oh(n \log n)$.
This completes a proof of the following theorem (corresponding to lines \ref{lst:line:find_slope_begin}-\ref{lst:line:find_slope_end} of the algorithm).

\begin{theorem}\label{app:th:binsearch}
	There exists an $\Oh(n \log n)$-time randomized algorithm returning $\lambda_1^l$ and $\lambda_1^r$ such that:
    \begin{compactenum}
    \item No intersections of boundary lines in $[\lambda_1^l, \lambda_1^r]$,
    \item There exists a solution $(\lambda_1^\dagger, \lambda_2^\dagger)$ such that $\lambda_1^\dagger \in [\lambda_1^l, \lambda_1^r]$. 
    \end{compactenum}
\end{theorem}

After we find $\lambda_1^l$ and $\lambda_1^r$ as in Theorem~\ref{app:th:binsearch} we show that there are only $\Oh(n)$ regions which can contain a solution and we can check them in $\Oh(n \log n)$ time.
The following theorem completes the proof of Theorem~\ref{thm:runtime} for $d = 2$:
\begin{theorem} \label{thm:find_region}
	If there exists a solution $\vlam^\dagger$ such that $\lambda_1^\dagger \in (\lambda_1^l; \lambda_1^r)$ and no intersection points are between $(\lambda_1^l; \lambda_1^r)$ then $\vlam^*$ can be found in $\Oh(n \log n)$ time.
\end{theorem}
\begin{proof}
	\begin{figure}[!h]
		\centering
		\begin{tikzpicture}[scale=1]
	\pgfmathsetmacro{\left}{-0.5}
	\pgfmathsetmacro{\right}{4}
	\pgfmathsetmacro{\bottom}{-0.5}
	\pgfmathsetmacro{\top}{4}
	\pgfmathsetmacro{\lb}{1}
	\pgfmathsetmacro{\rb}{3}
	\draw[thin,->,color=gray] (\left,0) -- (\right,0) node[right] {\small $\lambda_1$} coordinate(x axis);
	\draw[thin,->,color=gray] (0, \bottom) -- (0, \top) node[above] {\small $\lambda_2$} coordinate(y axis);
	
	\draw[thick,draw=red] (\lb, \bottom) -- (\lb, \top) node[above] {\small $\lambda_1^l$};
	\draw[thick,draw=red] (\rb, \bottom) -- (\rb, \top) node[above] {\small $\lambda_1^r$};
	
	\draw[name path = L3] (\left, 4) -- (\right, 2);
	\draw[name path = L2] (\left, 1) -- (\right, 2.5);
	\draw[name path = L1] (\left, 0) -- (\right, 1.5);
	
	\node at (2, 0.3) {$\mathtt{R_1}$};
	\node at (2, 1.3) {$\mathtt{R_2}$};
	\node at (2, 2.3) {$\mathtt{R_3}$};
	\node at (2, 3.3) {$\mathtt{R_4}$};
\end{tikzpicture}
		\caption{
			Final stage of Algorithm~\ref{alg:2D} (lines~\ref{lst:line:find_region_begin}-\ref{lst:line:find_region_end}), when there are no intersections of boundary lines between $\lambda_1^l$ and $\lambda_1^r$.
			Solution $\vlam^*$ belongs to one of $R_1, \ldots, R_4$, the sets resulting from partitioning of $(\lambda_1, \lambda_2) \times \mathbb{R}$ by boundary lines.
		}
		\label{fig:partition}
	\end{figure}
    We show how to find $\vlam^*$ in lines~\ref{lst:line:find_region_begin}-\ref{lst:line:find_region_end} of the algorithm.
	Consider set $S = (\lambda_1^l, \lambda_1^r) \times \mathbb R$.
	Let $\set{R_t}_{t=1}^T$ be the partition a of $S$ into parts lying between the boundary lines.
	Since $S$ doesn't contain boundary intersections and there are $\Oh(n)$ boundaries, the number of parts in the partition is $\Oh(n)$.
	For each $R_t$ we solve the following system of equations over $\lambda_1$ and $\lambda_2$:
	\[\begin{cases}
	\sum\limits_{i=1}^n \h 1_i(\lambda_1, \lambda_2) = c_1, \\
	\sum\limits_{i=1}^n \h 2_i(\lambda_1, \lambda_2) = c_2
	\end{cases}\]
	Since no boundary line crosses $R_t$, it is a subset of some region.
    Therefore, $\h 1_i$ and $\h 2_i$ are linear inside $R_t$, meaning that the above system becomes a system of linear equations.
	If the solution to the system belongs to $R_t$, then we can take it as $\vlam^*$.
	Thus it only remains to show how to find coefficients for the system in $\Oh(n \log n)$ total time.
	
	Recall that in Algorithm~\ref{alg:2D} we assume that $\set{R_t}$ are sorted from bottom to top.
	For $R_1$ we find the linear system coefficients in $\Oh(n)$ time.
	Assume that the $R_t$ are already computed.
	To find the coefficients for next set $R_{t+1}$, notice that $R_t$ and $R_{t+1}$ are separated by some boundary line.
	This line corresponds to some $\hj_i$ and therefore crossing it will change the coefficient of only this $\hj_i$, and the coefficients can be recomputed in $\Oh(1)$ time.
	Since there are $\Oh(n)$ boundary lines, the overall time for recomputation is also $\Oh(n)$.
	Taking sorting of $\set{R_t}$ into account, the total running time is $\Oh(n \log n)$.
\end{proof}
\section{Missing proofs from Section~\ref{SEC:PROJECTION}}
\label{app:proofs}
\begin{proof}[of Proposition~\ref{prop:approx-to-exact-reduction}]
The constraints corresponding to $j \in S_0$ are not tight for the correct guess, otherwise consider a guess which has appropriate signs corresponding to the tight constraints in the optimum solution. Let $\vx^*$ be the optimum without constraints for $j \in S_0$ and let $\vx^*_0$ be the optimum with these constraints. If these two optima are different then we can improve the optimum $\vx^*_0$ with the inequality constraints as follows. Consider vector $\vz =(1 - \alpha)\vx^*_0 + \alpha \vx^*$ for some small $\alpha > 0$. Because the constraints corresponding to $j \in S_0$ are not tight none of these constraints will be violated by this vector for small enough $\alpha$. All other constraints will be satisfied by convexity. However, we have $\|\vz - \vy\| < \|\vx^*_0 - \vy\|$, a contradiction with the optimality of $\vx^*_0$.

Uniqueness of the optimum follows from the uniqueness of projection on a convex body\dtodo{give reference}.
\end{proof}

\iffull
	
\fi

\bibliographystyle{abbrv}
\bibliography{vldb19}

\end{document}